\documentclass[3p,review]{elsarticle}

\usepackage{lineno,hyperref}
\usepackage{times}
\usepackage{type1cm}
\fontsize{12pt}{18pt}

\usepackage{amsmath}
\usepackage{amsfonts}
\usepackage{amssymb}
\usepackage{amsthm}
\usepackage{array}

\usepackage[caption=false,font=footnotesize,labelfont=rm,textfont=rm]{subfig}
\usepackage{overpic}
\usepackage{textcomp}
\usepackage{stfloats}
\usepackage{url}
\usepackage{verbatim}
\usepackage{graphicx}
\usepackage{epstopdf}
\usepackage{dcolumn}
\usepackage{bm}
\usepackage{color}
\usepackage{algorithmic}
\usepackage{algorithm}
\usepackage{float}
\usepackage{balance}
\usepackage{booktabs}
\usepackage{braket}
\usepackage{appendix}
\usepackage{multirow}
\usepackage{makecell}

\allowdisplaybreaks[4]
\modulolinenumbers[5]

\theoremstyle{plain}
\newtheorem{theorem}{Theorem}

\newtheorem{lemma}{Lemma}

\theoremstyle{definition}
\newtheorem{definition}{Definition}

\makeatletter
\newenvironment{breakablealgorithm}
  {
   \begin{center}
     \refstepcounter{algorithm}
     \hrule height.8pt depth0pt \kern2pt
     \renewcommand{\caption}[2][\relax]{
       {\raggedright\textbf{\ALG@name~\thealgorithm} ##2\par}%
       \ifx\relax##1\relax %
         \addcontentsline{loa}{algorithm}{\protect\numberline{\thealgorithm}##2}%
       \else %
         \addcontentsline{loa}{algorithm}{\protect\numberline{\thealgorithm}##1}%
       \fi
       \kern2pt\hrule\kern2pt
     }
  }{
     \kern2pt\hrule\relax
   \end{center}
  }
\makeatother

\journal{}

\date{}

\begin{document}
\begin{frontmatter}
\title{Distributed Exact Quantum Amplitude Amplification Algorithm for Arbitrary Quantum States}

\author[my1address]{Xu Zhou\corref{mycorrespondingauthor}}
\ead{zhouxu@tgqs.net}
\cortext[mycorrespondingauthor]{Corresponding author}

\author[my2address]{Wenxuan Tao}

\author[my3address]{Keren Li}

\author[my4address]{Shenggen Zheng}

\address[my1address]{Yangtze Delta Industrial Innovation Center of Quantum Science and Technology, Suzhou 215000, China}

\address[my2address]{Shanghai MiQro Era Digital Technology Co., Ltd., Shanghai 200023, China}

\address[my3address]{Shenzhen University, Shenzhen 518060, China}

\address[my4address]{Quantum Science Center of Guangdong-Hong Kong-Macao Greater Bay Area (Guangdong), Shenzhen 518045, China}

\begin{abstract}
In the noisy intermediate-scale quantum (NISQ) era, distributed quantum computation has garnered considerable interest, as it overcomes the physical limitations of single-device architectures and enables scalable quantum information processing. In this study, we focus on the challenge of achieving exact amplitude amplification for quantum states with arbitrary amplitude distributions and subsequently propose a Distributed Exact Quantum Amplitude Amplification Algorithm (DEQAAA). Specifically, (1) it supports partitioning across any number of nodes $t$ within the range $2 \leq t \leq n$; (2) the maximum qubit count required for any single node is expressed as $\max \left(n_0,n_1,\dots,n_{t-1} \right) $, where $n_j$ represents the number of qubits at the $j$-th node, with $\sum_{j=0}^{t-1} n_j =n$; (3) it can realize exact amplitude amplification for multiple targets of a quantum state with arbitrary amplitude distributions; (4) we verify the effectiveness of DEQAAA by resolving a specific exact amplitude amplification task involving two targets (8 and 14 in decimal) via MindSpore Quantum, a quantum simulation software, with tests conducted on 4-qubit, 6-qubit, 8-qubit and 10-qubit systems. Notably, through the decomposition of $C^{n-1}PS$ gates, DEQAAA demonstrates remarkable advantages in both quantum gate count and circuit depth as the qubit number scales, thereby boosting its noise resilience. In the 10-qubit scenario, for instance, it achieves a reduction of over $97\%$ in both indicators compared to QAAA and EQAAA, underscoring its outstanding resource-saving performance.
\end{abstract}

\begin{keyword}
Noisy intermediate-scale quantum (NISQ) era; Distributed quantum computation; Distributed Exact Quantum Amplitude Amplification Algorithm (DEQAAA); MindSpore Quantum; Quantum gate decomposition
\end{keyword}

\end{frontmatter}


\section{Introduction}\label{sec-introduction}
Quantum computing \cite{nielsen2010quantum} has emerged as a revolutionary paradigm in the field of computation, harnessing the principles of quantum mechanics to perform complex calculations at unprecedented speeds. The fundamental building blocks of quantum computing, qubits, exhibit unique properties such as superposition and entanglement, which enable quantum computers to process vast amounts of information simultaneously. This inherent parallelism has the potential to solve certain problems exponentially faster than classical computers \cite{deutsch1985quantum, deutsch1992rapid, bernstein1993quantum, simon1997power, shor1994algorithms, grover1997quantum, harrow2009quantum, peruzzo2014variational, farhi2014quantum}, making quantum computing a focal point of research and development in both academia and industry. 

Over the years, significant advancements have been made in the field of quantum computing, with milestones such as the demonstration of quantum supremacy \cite{arute2019quantum, zhong2020quantum, google2025quantum, gao2025establishing} highlighting the potential of quantum processors to outperform their classical counterparts. However, practical quantum computing still faces numerous challenges, particularly in the noisy intermediate-scale quantum (NISQ) era \cite{preskill2018quantum}. In this era, quantum devices are characterized by a limited number of qubits and high error rates, which restrict the complexity of quantum algorithms that can be implemented. These inherent limitations pose substantial challenges for achieving reliable and scalable quantum computation, yet they also drive innovative approaches to algorithm design.

To address these constraints and enable scalable quantum information processing, substantial research efforts have turned to distributed quantum computing \cite{barral2025review}. This paradigm employs multiple quantum nodes interconnected via quantum channels, effectively distributing quantum resources and alleviating the physical limitations of monolithic devices. By sharing the computational load across nodes, distributed systems can tackle problems of greater scale and complexity than any single device could manage. Its architecture is characterized by lower per-node qubit requirements and shallower circuit depths, thereby enhancing feasibility. Consequently, the development of quantum algorithms tailored for such distributed architectures has become a focal point of research. 

In the early years, there has been a significant amount of research work on algorithm design \cite{buhrman2003distributed, yimsiriwattana2004distributed, beals2013efficient} in distributed quantum computing. In 2021, Avron et al. \cite{avron2021quantum} introduced a technique for decomposing subfunctions from an original Boolean function $f: \{0,1\}^n \rightarrow \{0,1\}$. Inspired by this method, Qiu et al. \cite{qiu2024distributed} extended this method and developed serial and parallel versions of a distributed Grover's algorithm. While this constitutes an important advance, the parallel mode requires $2^k(n-k)$ qubits, where $2^k$ is the number of computing nodes. This substantial resource demand poses a significant scalability challenge for the scheme. To address the issues of excessive qubit usage, the limitation of node numbers to $2^k$, and the challenge of achieving exact search, Zhou et al. \cite{zhou2023distributedEGA, zhou2025distributedEGGA} successively proposed a single-target exact Grover's algorithm with $\lfloor n/2 \rfloor$ nodes and a multi-target exact generalized Grover's algorithm with  $2 \leq t \leq n$ nodes. Moreover, a series of distributed quantum algorithms have been successively proposed, including the distributed Deutsch-Jozsa (DJ) algorithm \cite{li2025distributedDJ}, the distributed Bernstein-Vazirani (BV) algorithm \cite{zhou2023distributedBV}, the distributed Simon's algorithm \cite{tan2022distributed, zhou2025distributedSimon}, and the distributed Shor's algorithm \cite{Xiao2022rjo}.

Significant strides have been made in the experimental realm of distributed quantum computing \cite{liu2021distributed, ying2023experimental, akhtar2023high}. In 2024, Liu et al. \cite{liu2024nonlocal} successfully implemented a non-local CNOT gate between two quantum nodes separated by 7 km, leveraging photonic quantum systems. Their approach, which included multimode solid-state quantum storage and a quantum gate teleportation protocol, enabled the execution of the DJ and quantum phase estimation algorithms. In 2025, Main et al. \cite{main2025distributed} demonstrated distributed quantum computation between two trapped-ion modules connected via photons, realizing a non-local CZ gate between qubits in independent modules. This work experimentally validated distributed quantum operations and successfully executed Grover's algorithm. Subsequently, Wei et al. \cite{wei2025universal} achieved distributed blind quantum computation on a solid-state platform, realizing a universal blind gate set and executing a hidden DJ algorithm across a two-node network. These advancements represent key milestones towards secure and scalable distributed quantum processing.

Quantum amplitude amplification algorithm (QAAA) \cite{brassard2002quantum} is a fundamental technique in quantum computing that significantly enhances the accuracy of quantum algorithms. At its core, amplitude amplification is a method to increase the amplitude of a desired quantum state, thereby improving the success probability of quantum algorithms. This technique is widely used in various quantum algorithms \cite{rajagopal2021quantum, mandl2024amplitude, perriello2025quantum}, including Grover's search algorithm \cite{grover1997quantum}, which provides a quadratic speedup for unstructured search problems compared to classical algorithms.

Despite its widespread application, QAAA still has certain limitations. First, traditional QAAA was not designed for distributed quantum computing architectures, which restricts its application in the NISQ era. Second, except in specific cases, standard QAAA generally cannot achieve exact amplitude amplification \cite{diao2010exactness}. These limitations consequently constrain the potential of QAAA within NISQ computers, necessitating the development of novel quantum algorithms to address these challenges.

Recently, Hua et al. \cite{hua2025distributed} introduced a distributed QAAA. In their work, they first posited that $\mathcal{A}=\mathcal{A}_1 \otimes \mathcal{A}_2 $. They then employed a subfunction partitioning technique analogous to that in Ref.~\cite{qiu2024distributed}, generating $2^k$ Boolean subfunctions with $(n-k)$-bit inputs. Subsequently, they utilized a fixed-point search algorithm \cite{yoder2014fixed} for each node, which is an enhancement of the original fixed-point search algorithm \cite{grover2005fixed} and operates without requiring prior knowledge of the target state probabilities. Finally, by determining the first $k$ bits based on the subfunction indices and the last $n-k$ bits based on the fixed-point search results, they achieved amplitude amplification of the original quantum state for specific $n$-qubit states.

However, their approach is subject to several limitations. Primarily, it relies on the assumption that the unitary operator $\mathcal{A}$ can be decomposed into a tensor product of two smaller unitaries, $\mathcal{A}_1$ and $\mathcal{A}_2 $. This assumption restricts its applicability to quantum states with arbitrary amplitude distributions. This constraint significantly limits the choice of initial states amenable to amplification in practical scenarios. Additionally, while the scheme reduces the qubit count per node from $n$ to $n-k$, it inherits the same drawbacks as Ref.~\cite{qiu2024distributed} regarding the total qubit resource overhead, which scales as $2^k(n-k)$, and the number of nodes, limited to $2^k$. Finally, akin to the original QAAA, this methodology generally cannot achieve exact amplitude amplification except under specific conditions.

In this paper, we present a novel Distributed Exact Quantum Amplitude Amplification Algorithm (DEQAAA) tailored to tackle the challenges inherent in quantum computing within the NISQ era. This algorithm is characterized by several notable features: (1) it supports distribution across an arbitrary number of nodes $t$ (with $2 \leq t \leq n$); (2) the maximum qubit requirement for any individual node is defined as $\max \left(n_0,n_1,\dots,n_{t-1} \right) $, where $n_j$ stands for the qubit count at the $j$-th node and $\sum_{j=0}^{t-1} n_j =n$; (3) it is able to realize exact amplitude amplification for multiple targets in a quantum state with arbitrary amplitude distributions.

We validate the proposed DEQAAA framework using MindSpore Quantum \cite{mq_2021, xu2024mindspore}, a comprehensive software library for quantum computing. Specifically, we demonstrate its efficacy by solving an exact amplitude amplification problem for two targets (8 and 14 in decimal), with tests carried out on 4-qubit, 6-qubit, 8-qubit and 10-qubit systems. The successful execution of the corresponding quantum circuits on this platform confirms the practical viability of our method. More importantly, the decomposition of multi-controlled $C^{n-1}PS$ gates enables DEQAAA to gain remarkable superiority in both quantum gate count and circuit depth with the scaling of qubit number, thus enhancing its noise robustness. This noise-resilient trait makes DEQAAA highly suitable for deployment in the NISQ era.

The structure of this paper is as follows. Following a review of the QAAA and its exact variant in Section \ref{sec-preliminary}, we formally introduce our DEQAAA in Section \ref{sec-DEQAA}. Section \ref{sec-analysis} then rigorously validates DEQAAA, proving its correctness and underscoring its performance benefits. The practical utility of our algorithm is demonstrated through an implementation on MindSpore Quantum in Section \ref{sec-experiment}. We conclude in Section \ref{sec-conclusion} by synthesizing our key contributions and proposing future research directions.

\section{Preliminary}\label{sec-preliminary}
In this section, we present a comprehensive overview of the QAAA \cite{brassard2002quantum} and its exact variant, with a focus on addressing the amplitude amplification problem.

\subsection{Quantum Amplitude Amplification Algorithm}\label{sec-QAA}
Consider an $n$-qubit state $\vert\Psi\rangle$ with an arbitrary amplitude distribution, generated by a unitary operator $\mathcal{A}$:
\begin{eqnarray}
\vert\Psi\rangle = \mathcal{A} \vert 0\rangle^{\otimes n} =  \sum_{x \in \{0,1\}^n} \gamma_x \vert x\rangle, 
\end{eqnarray}
where $\gamma_x$ are complex amplitudes satisfying $ \sum_{x \in \{0,1\}^n} \left \vert \gamma_x \right \vert ^2=1$.

A Boolean function $f:\{0,1\}^n\rightarrow\{0,1\}$ is given, which partitions the input space into two disjoint subsets:
\begin{equation}
f(x)=
\begin{cases}
1, & x \in X_g, \\
0, & x \in X_b,
\end{cases}
\label{eq-f-definition}
\end{equation}
where $X_g, X_b \in\{0,1\}^n$ represent the sets of target and non-target strings, respectively. Consequently,  $\vert\Psi\rangle$ can be expressed as
\begin{align}
\vert\Psi\rangle
&= \sum_{x \in X_g} \alpha_x \vert x\rangle + \sum_{x \in X_b} \beta_x \vert x\rangle\\
&= \sqrt{p_g}\vert\Psi_g\rangle + \sqrt{1-p_g}\vert\Psi_b\rangle.
\label{eq-psi-2part}
\end{align}
Here, $\alpha_x$ and $\beta_x$ denote the amplitudes of the basis state $\vert x\rangle$ belonging to the target and non-target sets, respectively. The component $\vert\Psi_g\rangle$ represents the superposition of all target states, while $\vert\Psi_b\rangle$ corresponds to the superposition of non-target states. The probability of obtaining a target state (success probability) is given by
\begin{equation}
p_g = \sum_{x \in X_g} \vert\alpha_x\vert^2 .
\label{eq-pg}
\end{equation}

The QAAA process, which resembles that of Grover's algorithm \cite{grover1997quantum}, involves iteratively applying the amplitude amplification operator $Q$ a total of $r$ times, where
\begin{equation}
    r=\left\lfloor\frac{\pi}{4\arcsin\left(\sqrt{p_g}\right)}\right\rfloor,
\label{eq-r-definition}
\end{equation}
to amplify the amplitude of the target states. The operator $Q$ is defined as
\begin{equation}
Q=\mathcal{A}S_{\vert 0 \rangle ^ {\otimes n}}\mathcal{A}^{\dagger}S_f,
\label{eq-Q-definition}
\end{equation}
where the operations of $S_f$ and $S_{\vert 0 \rangle ^ {\otimes n}}$ are defined as
\begin{equation}
\vert x\rangle\stackrel{S_f}{\longrightarrow}
\begin{cases}
-\vert x\rangle , &  x \in X_g , \\
\vert x\rangle , & x \in X_b ,
\end{cases}
\label{eq-Sf-2}
\end{equation}
and
\begin{equation}
\vert x\rangle\stackrel{S_{\vert 0 \rangle ^ {\otimes n}}}{\longrightarrow}
\begin{cases}
-\vert 0\rangle^{\otimes n} , & \vert x\rangle=\vert 0\rangle^{\otimes n} , \\
\vert x\rangle , & \text{otherwise,}
\end{cases}
\label{eq-S0-2}
\end{equation}
respectively.

We provide a concise overview of the QAAA in Algorithm \ref{algo-QAA}, with a schematic of the quantum circuit shown in Figure \ref{fig-QAA}. A detailed analysis of the QAAA is provided in Appendix \ref{AnalysisQAAA}.

\begin{breakablealgorithm}
\caption{Quantum Amplitude Amplification Algorithm}
\label{algo-QAA}
\renewcommand{\algorithmicrequire}{\textbf{Input:}} 
\renewcommand{\algorithmicensure}{\textbf{Output:}}
\begin{algorithmic}[1]
\REQUIRE{
\makebox[1.0em][r]{(1)} \parbox[t]{0.93\linewidth}{The number of qubits $n$;} \\
\makebox[2.5em][r]{(2)} \parbox[t]{0.93\linewidth}{A unitary operator $\mathcal{A}$ that prepares the initial state $\vert\Psi\rangle$;} \\
\makebox[2.5em][r]{(3)} \parbox[t]{0.93\linewidth}{A Boolean function $f: \{0,1\}^n \rightarrow \{0,1\}$ identifying the target set $X_g$ (i.e., $f(x)=1$ for $x\in X_g$, and $0$ otherwise);} \\ [0.8em]
\makebox[2.5em][r]{(4)} \parbox[t]{0.93\linewidth}{The success probability $p_g$ as in Eq.~\eqref{eq-pg};} \\
\makebox[2.5em][r]{(5)} \parbox[t]{0.93\linewidth}{The number of iterations $r$ of the operator $Q$ as in Eq.~\eqref{eq-r-definition};} \\
\makebox[2.5em][r]{(6)} \parbox[t]{0.93\linewidth}{The amplitude amplification operator $Q$ as in Eq.~\eqref{eq-Q-definition}, along with the component operators $S_f$ and $S_{\vert 0 \rangle ^ {\otimes n}}$ as in Eq.~\eqref{eq-Sf-2} and Eq.~\eqref{eq-S0-2}, respectively.}\\[0.8em]
}

\ENSURE{The target string $x\in X_g$ with great probability.}

\STATE Initialize $n$ qubits as $\vert 0\rangle^{\otimes n}$.
\STATE Apply the unitary operator $\mathcal{A}$ once.
\STATE Initialize $k = 0$.

\WHILE{$k < r$}
    \STATE Apply the operator $Q$ once.
    \STATE Update $k = k + 1$.
\ENDWHILE

\STATE Measure each qubit in the basis $\{\vert 0\rangle, \vert 1\rangle\}$.
\end{algorithmic}
\end{breakablealgorithm}

\begin{figure}[H]
\centering
\includegraphics[width=0.4\textwidth]{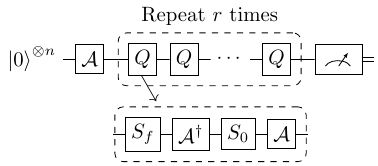}
\caption{Quantum circuit of the QAAA.}
\label{fig-QAA}
\end{figure}

\subsection{Exact Quantum Amplitude Amplification Algorithm}\label{sec-EQAA}
The QAAA can output $x\in X_g$ with a high probability of success. However, except in special cases, it does not achieve exact precision, which would require a 100\% probability of outputting $x\in X_g$.

Inspired by the exact Grover's algorithm by Long \cite{long2001grover,zhou2023distributedEGA, zhou2025distributedEGGA}, the key concept involves extending the Grover operator $G$ to the exact search operator $L$. Similarly, the QAAA can be enhanced by replacing the operator $Q$ with an exact amplitude amplification operator $EQ$. This modification, built on the foundation of QAAA, enables exact amplification and leads to the Exact Quantum Amplitude Amplification Algorithm (EQAAA).

The operator $EQ$ retains the composition structure of $Q$, but incorporates new operators $R_{f}^{\phi}$ and $R_{\vert 0 \rangle ^ {\otimes n}} ^{\phi}$, defined as follows:
\begin{equation}
    EQ = \mathcal{A} R_{\vert 0 \rangle ^ {\otimes n}} ^{\phi} \mathcal{A}^{\dagger} R_{f}^{\phi},
    \label{eq-EQ}
\end{equation}
where the operations of $R_{f}^{\phi}$ and $R_{\vert 0 \rangle ^ {\otimes n}} ^{\phi}$ are defined as
\begin{equation}
\vert x\rangle\stackrel{R_{f}^{\phi}}{\longrightarrow}
\begin{cases}
e^{i \phi} \vert x\rangle , & x \in X_g , \\
\vert x\rangle , & x \in X_b,
\end{cases}
\label{eq-Rf-2}
\end{equation}
and
\begin{equation}
\vert x\rangle\stackrel{R_{\vert 0 \rangle ^ {\otimes n}} ^{\phi}}{\longrightarrow}
\begin{cases}
e^{i \phi} \vert 0\rangle^{\otimes n} , & \vert x\rangle=\vert 0\rangle^{\otimes n} , \\
\vert x\rangle , & \text{otherwise,}
\end{cases}
\label{eq-R0-2}
\end{equation}
respectively.

Similarly, the operator $EQ$ will be iteratively applied $J + 1$ times, where
\begin{equation}
    J = \left\lfloor \frac{\pi}{4\arcsin\left(\sqrt{p_g}\right)} - \frac{1}{2} \right\rfloor,
    \label{eq-J}
\end{equation}
and $p_g$ is defined by Eq.~\eqref{eq-pg}.

The phase angle $\phi$ involved in $R_{f}^{\phi}$ and $R_{\vert 0 \rangle ^ {\otimes n}} ^{\phi}$ is calculated as follows:
\begin{equation}
    \phi = 2 \arcsin \left(\frac{\sin \left(\frac{\pi}{4 J + 6}\right)}{\sqrt{p_g}}\right).
    \label{eq-angle-phi}
\end{equation}

We present a concise summary of the EQAAA in Algorithm \ref{algo-EQAA}, accompanied by a schematic illustration of the quantum circuit in Figure \ref{fig-EQAA}. A comprehensive analysis of the EQAAA is detailed in Appendix \ref{AnalysisEQAAA}.

\begin{breakablealgorithm}
\caption{Exact Quantum Amplitude Amplification Algorithm}
\label{algo-EQAA}
\renewcommand{\algorithmicrequire}{\textbf{Input:}} 
\renewcommand{\algorithmicensure}{\textbf{Output:}}
\begin{algorithmic}[1]
\REQUIRE{
\makebox[1.0em][r]{(1)} \parbox[t]{0.93\linewidth}{The number of qubits $n$;} \\
\makebox[2.5em][r]{(2)} \parbox[t]{0.93\linewidth}{A unitary operator $\mathcal{A}$ that prepares the initial state $\vert\Psi\rangle$;} \\
\makebox[2.5em][r]{(3)} \parbox[t]{0.93\linewidth}{A Boolean function $f: \{0,1\}^n \rightarrow \{0,1\}$ identifying the target set $X_g$ (i.e., $f(x)=1$ for $x\in X_g$, and $0$ otherwise);} \\ [0.8em]
\makebox[2.5em][r]{(4)} \parbox[t]{0.93\linewidth}{The success probability $p_g$ as in Eq.~\eqref{eq-pg};} \\
\makebox[2.5em][r]{(5)} \parbox[t]{0.93\linewidth}{The exact amplitude amplification operator $EQ$ as in Eq.~\eqref{eq-EQ}, along with the component operators $R_{f}^{\phi}$ and $R_{\vert 0 \rangle ^ {\otimes n}} ^{\phi}$ as in Eq.~\eqref{eq-Rf-2} and Eq.~\eqref{eq-R0-2}, respectively;} \\ [0.4em]
\makebox[2.5em][r]{(6)} \parbox[t]{0.93\linewidth}{The number of iterations $J+1$ of the exact operator $EQ$, where $J$ is defined as in Eq.~\eqref{eq-J};} \\
\makebox[2.5em][r]{(7)} \parbox[t]{0.93\linewidth}{The phase angle $\phi$ as in Eq.~\eqref{eq-angle-phi}.}
}
\ENSURE{The target string $x\in X_g$ with a probability of $100 \%$.}

\STATE Initialize $n$ qubits as $\vert 0\rangle^{\otimes n}$.
\STATE Apply the unitary operator $\mathcal{A}$ once.
\STATE Initialize $k = 0$.

\WHILE{$k < J + 1$}
    \STATE Apply the operator $EQ$ once.
    \STATE Update $k = k + 1$.
\ENDWHILE

\STATE Measure each qubit in the basis $\{\vert 0\rangle, \vert 1\rangle\}$.
\end{algorithmic}
\end{breakablealgorithm}

\begin{figure}[H]
\centering
\includegraphics[width=0.4\textwidth]{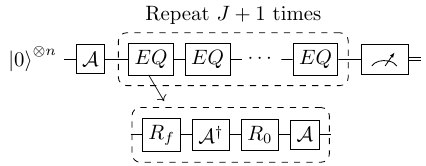}
\caption{Quantum circuit of the EQAAA.}
\label{fig-EQAA}
\end{figure}

\section{Distributed Exact Quantum Amplitude Amplification Algorithm}\label{sec-DEQAA}

We present the Distributed Exact Quantum Amplitude Amplification Algorithm (DEQAAA) in this section. It operates on a distributed quantum system composed of $ t $ nodes, with $2 \le t \le n$.

\subsection{Substates}
The successful application of both QAAA and EQAAA relies on specific prior knowledge: the probability distribution of the initial state $\vert\Psi\rangle$ and the number of target strings in the Boolean function $f(x)$. This knowledge allows for the calculation of the success probability $p_g$, the core parameter that dictates the entire amplification process. Accordingly, $p_g$ directly specifies the number of iterations required by each algorithm: $r$ for QAAA and $J+1$ for EQAAA. It is imperative to note that a precise determination of the probability distribution of $\vert\Psi\rangle$ is a necessary condition for the exact implementation of EQAAA.

Similarly, an analogous assumption is required: we presume that the probability distribution of an arbitrary quantum state can be precisely determined. This allows the success probability, $p_g$, to be defined according to the number of target strings specified by $f(x)$. In practice, whether in numerical simulations or physical experiments, this distribution is approximated statistically via repeated measurements. The accuracy of the experimental distribution in approximating the theoretical one improves with the number of measurement shots. This convergence will be quantitatively analyzed in Section \ref{sec-experiment}.

As outlined in Section \ref{sec-preliminary}, an $n$-qubit state $ \vert \Psi\rangle$ with an arbitrary amplitude distribution is prepared by applying an arbitrary unitary operator $\mathcal{A}$ to the initial state:
\begin{eqnarray}
\vert\Psi\rangle = \mathcal{A} \vert 0\rangle^{\otimes n} =  \sum_{x \in \{0,1\}^n} \gamma_x \vert x\rangle.
\end{eqnarray}
Based on the aforementioned assumption, we denote the exact probability distribution of $\vert \Psi\rangle$ as $P$, specifying the measurement probability of $\vert x\rangle$ as 
\begin{eqnarray}
P(x)=  \vert \gamma_x \vert ^2.
\end{eqnarray}

Consider a distributed system composed of $2 \leq t \leq n$ nodes with a total of $n$ qubits. The next step involves deriving the local probability distribution $P_j$ for each individual node from the global distribution $P$. A substate, denoted as $\vert\varphi_j\rangle$, is then determined for the $j$-th node (where $j \in \{ 0, 1, \cdots, t-1 \}$) according to its specific distribution $P_j$. It is crucial to clarify that the mapping from a probability distribution to a quantum state is not unique; numerous distinct substates are consistent with the same $P_j$. For the purposes of the DEQAAA, any substate that satisfies $P_j$ is acceptable for implementing the DEQAAA. The detailed steps are outlined in Algorithm \ref{algo-substate} below.

\begin{breakablealgorithm}
\caption{The determination algorithm of the substates.}
\label{algo-substate}
\renewcommand{\algorithmicrequire}{\textbf{Input:}} 
\renewcommand{\algorithmicensure}{\textbf{Output:}}
\begin{algorithmic}[1]
\REQUIRE{
\makebox[1.0em][r]{(1)} \parbox[t]{0.93\linewidth}{The total number of qubits $n$;} \\
\makebox[2.5em][r]{(2)} \parbox[t]{0.93\linewidth}{The number of computing nodes $2 \leq t \leq n$;} \\
\makebox[2.5em][r]{(3)} \parbox[t]{0.93\linewidth}{The number of qubits of the $j$-th node, denoted by $n_j$, where $j \in \{ 0, 1, \cdots, t-1 \}$ and $\sum_{j=0}^{t-1} n_j = n$;} \\
\makebox[2.5em][r]{(4)} \parbox[t]{0.93\linewidth}{The exact probability distribution $P$ of $\vert\Psi\rangle$.}
}
\ENSURE{The substates $\vert \varphi_j \rangle$, where $j \in \{ 0, 1, \cdots, t-1 \}$.}

\STATE Initialize $j = 0$.
\WHILE{$j < t$}
    \IF{$j = 0$}
        \STATE Let $\delta = 0$, where $\delta$ represents the sum of the qubits preceding the $j$-th node.
    \ELSE
        \STATE Let $\delta = \sum_{a=0}^{j-1} n_a$.
    \ENDIF
    \STATE Let $\sigma = n - \delta - n_j$, where $\sigma$ represents the sum of the qubits following the $j$-th node.
    \STATE Initialize $k = 0$.
    \WHILE{$k < 2^{n_j}$}
        \STATE Let
        \begin{equation}
            P_j\left(\text{bin}\left(k,n_j\right)\right) = \sum_{u = 0}^{2^\delta} \sum_{v = 0}^{2^{\sigma}} P\left(\text{bin}\left(u \cdot 2^{n - \delta} + k \cdot 2^{\sigma} + v,n\right)\right),
        \end{equation}
        where $P_j$ is the probability distribution of the $j$-th node, and $\text{bin}(y,z)$ transfers a decimal number $y$ to a $z$-bit binary string.
        \STATE Determine a substate according to $P_j$:
        \begin{eqnarray}
            \vert\varphi_j\rangle=\sum_{x\in\{0,1\}^{n_j}}\sqrt{P_j(x)}\vert x\rangle.
        \end{eqnarray}
        \STATE Update $k = k + 1$.
    \ENDWHILE
    \STATE Update $j = j + 1$.
\ENDWHILE
\end{algorithmic}
\end{breakablealgorithm}

Once the substate $\vert\varphi_j\rangle$ is determined, it implies the existence of an $n_j$-qubit unitary operator $\mathcal{A}_{\vert \varphi_j \rangle}$ that can generate this substate:
\begin{eqnarray}
\vert\varphi_j\rangle = \mathcal{A}_{\vert \varphi_j \rangle}  \vert 0\rangle^{\otimes n_j}.
\end{eqnarray}
In practice, any quantum circuit that constructs the substate $\vert \varphi_j \rangle$ from the initial state $ \vert 0\rangle^{\otimes n_j}$ can be considered a valid realization of the operator $\mathcal{A}_{\vert \varphi_j \rangle}$. Established techniques for this purpose include Top-Down amplitude encoding \cite{araujo2023configurable}, uniformly controlled rotations \cite{Mottonen2005}, and Quantum Generative Adversarial Networks (QGANs) \cite{zoufal2019quantum}.

\subsection{DEQAAA with $2\leq t\leq n$ nodes}
Having defined the substate for each node, we now specify the corresponding Boolean subfunction $f_j$ and its set of target strings $X_j$ for each node $j$. These are derived from the original global Boolean function $f(x)$ (Eq.~\ref{eq-f-definition}) and its target strings set $X_g$. 

Without loss of generality, let $X_g = \{ x^{(0)}, x^{(1)}, \dots, x^{(\vert X_g\vert-1)} \} $ denote the set of $\vert X_g\vert$ target strings for $f(x)$, where $\vert X_g\vert$ is the number of target strings in the global target set $X_g$. Each string $ x^{(k)} $ is an $n$-bit string expressed as 
\begin{eqnarray}
x^{(k)} = x _{0}^{(k)}x_{1}^{(k)}\cdots x_{n-1}^{(k)},
\end{eqnarray}
where $x_{i}^{(k)}\in\{0,1\}$. Here, $i\in\{0,1,\cdots,n-1\}$ and $k\in\{0,1,\cdots, \vert X_g\vert-1\}$.

In order to extract the target substring, we define the function $h$ as follows:
\begin{eqnarray}
h(x, s, l) = x_s x_{s+1} \ldots x_{s+l-1},
\end{eqnarray}
where $x = x_0 x_1 \ldots x_{n-1}$ is an $n$-bit binary string, $s$ is the starting index, and $l$ is the length of the substring. The constraints are $0 \leq s \leq n-l$ and $ 1 \leq l \leq n$.

In a distributed system of $t$ nodes ($2 \leq t \leq n$), the $j$-th node is allocated $n_j$ qubits. The total number of qubits across all nodes is $n$, i.e. $\sum_{j=0}^{t-1} n_j = n$. Denoting by $S_j = \sum_{a=0}^{j-1} n_a$ the cumulative qubit count for nodes preceding $j$ (with $S_0 = 0$), the local target set for node $j$ is
\begin{equation}
X_j = \left \{ h\left (x^{(i)}, S_j, n_j\right) \mid i = 0, 1, \cdots, \vert X_g\vert-1 \right\}.
\label{eq-target-set}
\end{equation}
The corresponding Boolean subfunction $f_j$ for the $j$-th node is then defined as:
\begin{equation}
f_j(x) = 
\begin{cases} 
1, & \text{if } x \in X_j, \\ 
0, & \text{otherwise}.
\end{cases}
\label{fjx}
\end{equation}
Consequently, the probability of success on this node, denoted $p_j$, is given by
\begin{equation}
p_j = \sum_{x \in X_j} P_j(x).
\label{eq-pj}
\end{equation}

The DEQAAA operates through two phases: initial parallel execution of local EQAAA across all nodes to amplify the amplitudes of candidate target basis states within their respective subspaces, followed by a global EQAAA that acts on the entire system to amplify the amplitudes of the definitive global target states.

The first phase also requires defining the $EQ_j$ operator for the $j$-th node to execute the EQAAA. The operator is defined as follows: 
\begin{equation}
    EQ_j = \mathcal{A}_{\vert \varphi_j \rangle} R_{\vert 0 \rangle ^ {\otimes n_j}} ^{\phi_j} \mathcal{A}_{\vert \varphi_j \rangle}^{\dagger} R_{f_j}^{\phi_j},
    \label{eq-EQ-j}
\end{equation}
where the operations of $R_{f_j}^{\phi_j}$ and $R_{\vert 0 \rangle ^ {\otimes n_j}} ^{\phi_j}$ are defined as
\begin{equation}
\vert x\rangle\stackrel{R_{f_j}^{\phi_j}}{\longrightarrow}
\begin{cases}
e^{i \phi_j} \vert x\rangle , & x \in X_j , \\
\vert x\rangle , & \text{otherwise},
\end{cases}
\label{eq-Rfj}
\end{equation}
and
\begin{equation}
\vert x\rangle\stackrel{R_{\vert 0 \rangle ^ {\otimes n_j}} ^{\phi_j}}{\longrightarrow}
\begin{cases}
e^{i \phi_j} \vert 0\rangle^{\otimes n_j} , & \vert x\rangle=\vert 0\rangle^{\otimes n_j} , \\
\vert x\rangle , & \text{otherwise,}
\end{cases}
\label{eq-R0j}
\end{equation}
respectively.

The operator $EQ_j$ will be iteratively applied $J_j + 1$ times, where
\begin{equation}
    J_j = \left\lfloor \frac{\pi}{4\arcsin\left(\sqrt{p_j}\right)} - \frac{1}{2} \right\rfloor.
    \label{eq-J-j}
\end{equation}

The phase angle $\phi_j$ involved in $R_{f_j}^{\phi_j}$ and $R_{\vert 0 \rangle ^ {\otimes n_j}} ^{\phi_j}$ is calculated as follows:
\begin{equation}
    \phi_j = 2 \arcsin \left(\frac{\sin \left(\frac{\pi}{4 J_j + 6}\right)}{\sqrt{p_j}}\right).
    \label{eq-angle-phi-j}
\end{equation}

Following the first phase of DEQAAA, the global system state evolves from $\vert \Psi \rangle$ to $\vert \Psi_1 \rangle$. The probability distribution $P'$ associated with $\vert \Psi_1 \rangle$ can be precisely determined based on the previous assumptions. The updated success probability, $p'_g$, is therefore defined by
\begin{equation}
p'_g= \sum_{x \in X_g} P'(x) .
\label{eq-p'g}
\end{equation}

If the condition 
\begin{eqnarray}
p'_g=1 
\end{eqnarray}
is satisfied, it indicates that the exact amplitude amplification for the original $\vert \Psi \rangle$ has been successfully accomplished through the first phase alone. Otherwise, the system proceeds to execute the second phase of the DEQAAA.

In the second phase, the entire unitary operation of the first phase is encapsulated in a new composite operator
\begin{eqnarray}
\mathcal{B} = \left[ \bigotimes_{j=0}^{t-1} \left( EQ_{j}^{J_{j}+1} \right)  \right]  \mathcal{A},
\label{unitaryB}
\end{eqnarray}
where $EQ_{j}$ and $J_{j}$ are the operators previously defined in Eq.~\eqref{eq-EQ-j} and Eq.~\eqref{eq-J-j}, respectively. A global EQAAA is then executed over the entire system,  which drives the system to the final state $\vert \Psi_2 \rangle$ by amplifying the amplitudes of the global target basis states.

Similarly, the global exact amplitude amplification operator $\widehat{EQ}$ is defined as follows: 
\begin{equation}
    \widehat{EQ}= \mathcal{B} R_{\vert 0 \rangle ^{\otimes n}} ^{\hat{\phi}} \mathcal{B}^{\dagger} R_{f}^{\hat{\phi}},
    \label{eq-EQ-hat}
\end{equation}
where the operations of $R_{f}^{\hat{\phi}}$ and $R_{\vert 0 \rangle ^{\otimes n}} ^{\hat{\phi}}$ are defined as
\begin{equation}
\vert x\rangle\stackrel{R_{f}^{\hat{\phi}}}{\longrightarrow}
\begin{cases}
e^{i \hat{\phi}} \vert x\rangle , & x \in X_g , \\
\vert x\rangle , & \text{otherwise},
\end{cases}
\label{eq-Rf-hat}
\end{equation}
and
\begin{equation}
\vert x\rangle\stackrel{R_{\vert 0 \rangle ^{\otimes n}} ^{\hat{\phi}}}{\longrightarrow}
\begin{cases}
e^{i \hat{\phi}} \vert 0\rangle^{\otimes n} , & \vert x\rangle=\vert 0\rangle^{\otimes n} , \\
\vert x\rangle , & \text{otherwise,}
\end{cases}
\label{eq-R0-hat}
\end{equation}
respectively. 

The operator $\widehat{EQ}$ will be iteratively applied $\hat{J}+1$ times, where
\begin{equation}
    \hat{J}= \left\lfloor \frac{\pi}{4\arcsin\left(\sqrt{p'_g}\right)} - \frac{1}{2} \right\rfloor,
    \label{eq-J-hat}
\end{equation}
and $p'_g$ is defined by Eq.~\eqref{eq-p'g}. 

The phase angle $\hat{\phi}$ involved in $R_{f}^{\hat{\phi}}$ and $R_{\vert 0 \rangle ^{\otimes n}} ^{\hat{\phi}}$ is calculated as follows:
\begin{equation}
    \hat{\phi} = 2 \arcsin \left(\frac{\sin \left(\frac{\pi}{4 \hat{J}+ 6}\right)}{\sqrt{p'_g}}\right).
    \label{eq-angle-phi-hat}
\end{equation}

Through this two-phase procedure, the DEQAAA achieves exact amplitude amplification of the original target states. An outline of the DEQAAA is presented in Algorithm~\ref{algo-DEQAA}, and Figure~\ref{fig-DEQAA} depicts the corresponding quantum circuit.

\begin{breakablealgorithm}
\caption{Distributed Exact Quantum Amplitude Amplification Algorithm with $2 \leq t \leq n$ computing nodes.}
\label{algo-DEQAA}
\renewcommand{\algorithmicrequire}{\textbf{Input:}} 
\renewcommand{\algorithmicensure}{\textbf{Output:}}
\begin{algorithmic}[1]
\REQUIRE{
\makebox[1.0em][r]{(1)} \parbox[t]{0.93\linewidth}{The total number of qubits $n$;} \\
\makebox[2.5em][r]{(2)} \parbox[t]{0.93\linewidth}{A unitary operator $\mathcal{A}$ that prepares the initial state $\vert\Psi\rangle$;} \\
\makebox[2.5em][r]{(3)} \parbox[t]{0.93\linewidth}{The exact probability distribution of $ \vert \Psi\rangle$ as $P$;} \\
\makebox[2.5em][r]{(4)} \parbox[t]{0.93\linewidth}{A Boolean function $f: \{0,1\}^n \rightarrow \{0,1\}$ identifying the target set $X_g$ (i.e., $f(x)=1$ for $x\in X_g$, and $0$ otherwise);} \\[0.8em]
\makebox[2.5em][r]{(5)} \parbox[t]{0.93\linewidth}{The number of computing nodes $2 \leq t \leq n$;} \\
\makebox[2.5em][r]{(6)} \parbox[t]{0.93\linewidth}{The number of qubits of the $j$-th node, denoted by $n_j$, where $j \in \{ 0, 1, \cdots, t-1 \}$ and $\sum_{j=0}^{t-1} n_j = n$;} \\
\makebox[2.5em][r]{(7)} \parbox[t]{0.93\linewidth}{The substate $\vert \varphi_j \rangle$ and its preparation unitary operator $\mathcal{A}_{\vert \varphi_j \rangle}$ according to Algorithm \ref{algo-substate};} \\
\makebox[2.5em][r]{(8)} \parbox[t]{0.93\linewidth}{The Boolean subfunction $f_j$ as in Eq.~\eqref{fjx};} \\
\makebox[2.5em][r]{(9)} \parbox[t]{0.93\linewidth}{The success probability $p_j$ as in Eq.~\eqref{eq-pj};} \\
\makebox[2.5em][r]{(10)} \parbox[t]{0.93\linewidth}{The local exact amplitude amplification operator $EQ_j$ as in Eq.~\eqref{eq-EQ-j}, along with the component operators $R_{f_j}^{\phi_j}$ and $R_{\vert 0 \rangle ^ {\otimes n_j}} ^{\phi_j}$ as in Eq.~\eqref{eq-Rfj} and Eq.~\eqref{eq-R0j}, respectively;} \\[0.4em]
\makebox[2.5em][r]{(11)} \parbox[t]{0.93\linewidth}{The number of iterations $J_j+1$ of the exact operator $EQ_j$, where $J_j$ is defined as in Eq.~\eqref{eq-J-j};} \\
\makebox[2.5em][r]{(12)} \parbox[t]{0.93\linewidth}{The phase angle $\phi_j$ as in Eq.~\eqref{eq-angle-phi-j};}\\
\makebox[2.5em][r]{(13)} \parbox[t]{0.93\linewidth}{The entire unitary operation $\mathcal{B}$ of the first phase as in Eq.~\eqref{unitaryB}.}
}

\ENSURE{The target string $x\in X_g$ with a probability of $100 \%$.}

\STATE Initialize $n$ qubits as $\vert 0\rangle^{\otimes n}$.
\STATE Apply the unitary operator $\mathcal{A}$ once and get $\vert \Psi \rangle$.
\STATE Initialize $j = 0$.

\WHILE{$j < t$}
    \STATE Initialize $k = 0$.
    \WHILE{$k < J_j + 1$}
        \STATE Apply the operator $EQ_j$ once.
        \STATE Update $k = k + 1$.
    \ENDWHILE
    \STATE Update $j = j + 1$.
\ENDWHILE

\STATE Obtain the exact probability distribution $P'$ of the current quantum state $\vert \Psi_1 \rangle$ and calculate the success probability $p'_g$, as in Eq.~\eqref{eq-p'g}.

\IF{$p'_g=1$}
    \STATE Jump to the final step;
\ELSE
    \STATE Proceed to the next step.
\ENDIF
\STATE  Define the global exact amplitude amplification operator $\widehat{EQ}$ as in Eq.~\eqref{eq-EQ-hat}.
\STATE Calculate the number of iterations $\hat{J}+1$ as in Eq.~\eqref{eq-J-hat}, where $\hat{\phi}$  is calculate as in Eq.~\eqref{eq-angle-phi-hat}.

\STATE Initialize $q = 0$.

\WHILE{$q < \hat{J} + 1$}
    \STATE Apply the operator $\widehat{EQ}$ once.
    \STATE Update $q = q + 1$.
\ENDWHILE

\STATE Measure each qubit of $\vert \Psi_2 \rangle$ in the basis $\{\vert 0\rangle, \vert 1\rangle\}$.
\end{algorithmic}
\end{breakablealgorithm}

\begin{figure}[H]
\centering
\includegraphics[width=0.97\textwidth]{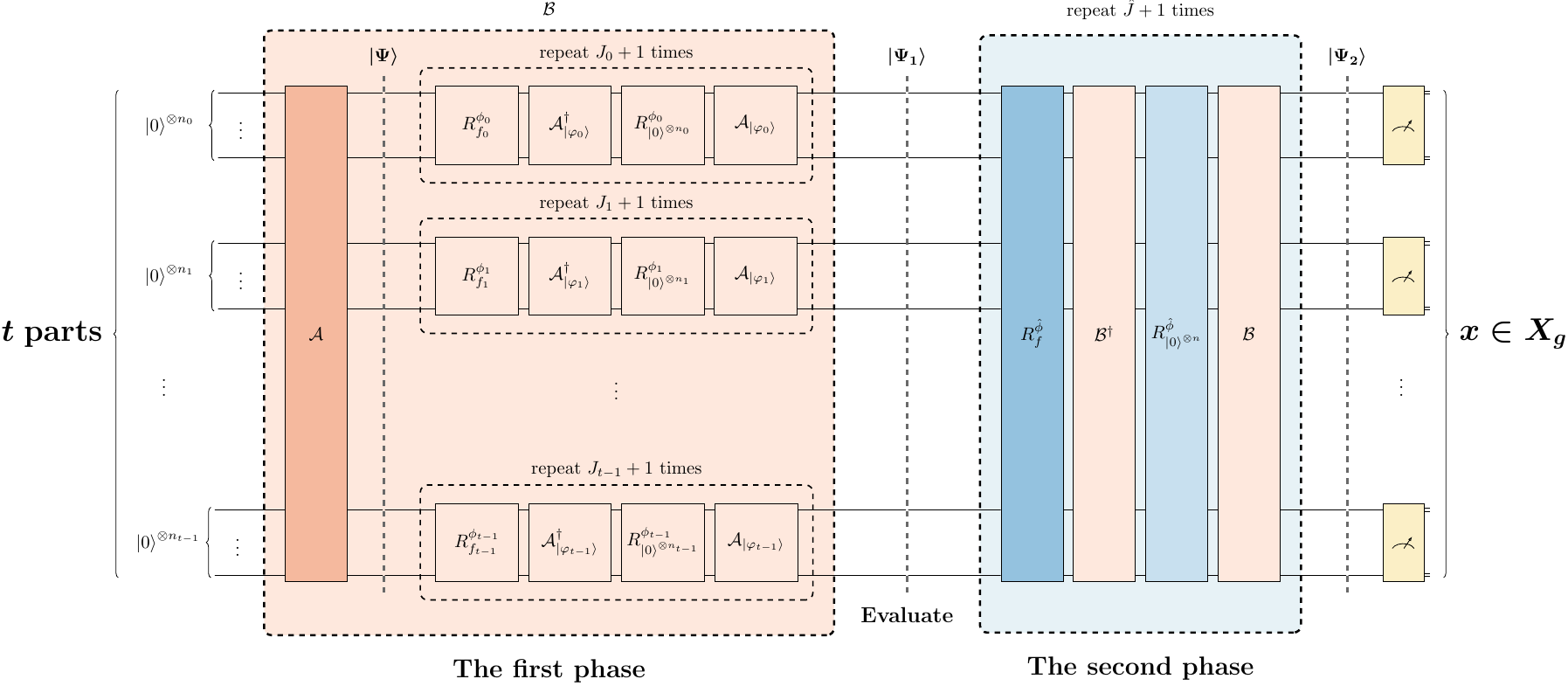}
\caption{Quantum circuit of the DEQAAA.}
\label{fig-DEQAA}
\end{figure}

\section{Analysis}\label{sec-analysis}
Our analysis comprehensively assesses the DEQAAA through three key aspects. We first validate the algorithmic correctness to ensure its theoretical soundness. We then rigorously analyze the circuit depth to quantify the advantage conferred by its distributed architecture. Finally, we benchmark the performance of our scheme against relevant centralized (QAAA, EQAAA) and distributed (existing DQAAA) counterparts.

\subsection{Correctness}
The criterion for validating the correctness of the DEQAAA is demonstrating that Algorithm \ref{algo-DEQAA} achieves exact amplitude amplification for all target strings $x$ belonging to the global target set $X_g$ culminating in a theoretical success probability of 1. This guarantee of deterministic success is formally established in Theorem \ref{proofcorrectness} below.

\begin{theorem}\label{proofcorrectness}
(\textbf{Correctness of DEQAAA}) 
Let $\vert\Psi\rangle = \mathcal{A}\vert0\rangle^{\otimes n}$ be an $n$-qubit state with an arbitrary amplitude distribution, where $\mathcal{A}$ denotes a unitary state-preparation operator. 
Let $f: \{0,1\}^n \rightarrow \{0,1\}$ be a Boolean function that identifies the target computational basis states within $\vert\Psi\rangle$ via the set $X_g = \{x\in\{0,1\}^n \vert f(x)=1\}$. 
The DEQAAA (Algorithm~\ref{algo-DEQAA}) achieves exact amplitude amplification for the target strings $x \in X_g$. Specifically, the final state $\vert \Psi_2\rangle$ satisfies:
\begin{equation}
\sum_{x \in X_g} \vert\langle x \vert \Psi_2 \rangle \vert^2 = 1,
\end{equation}
which guarantees that a measurement of $\vert\Psi_2\rangle$ yields a target string $x \in X_g$ with probability 1. 
\end{theorem}

\begin{proof}
The detailed proof is shown in Appendix \ref{prooftheorem1}. 
\end{proof}

\subsection{Circuit depth}\label{circuitdepth}
To quantitatively demonstrate the advantages inherent in distributed quantum algorithms, we first provide a formal definition of quantum circuit depth. We then rigorously characterize the circuit depths of QAAA, EQAAA and DEQAAA, thus establishing a foundational metric for assessing the efficiency of DEQAAA relative to its centralized counterparts.

\begin{definition} \label{defdepth}
(\textbf{Depth of Circuit} \cite{zhou2023distributedEGA, zhou2025distributedEGGA}) The depth of a quantum circuit is defined as the number of time steps required to execute it, corresponding to the length of the longest continuous path of sequential quantum operations from the initial to the final state. 
\end{definition}

For example, the left circuit in Figure \ref{fig-depth}, with a depth of 1, represents a single operation. In contrast, the right circuit has a depth of 7, reflecting a longer sequential operation sequence, yet both are functionally equivalent.

\begin{figure}[H]
\centering
\includegraphics[width=0.4\linewidth]{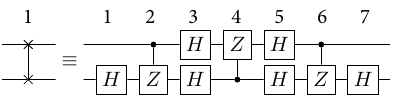}
\vspace{-1.5em}
\caption{Two quantum circuits implementing the SWAP gate.}
\label{fig-depth}
\end{figure}

Actually, the depth of a quantum circuit is contingent upon the arrangement of its constituent quantum gates. A critical step in its analysis, therefore, involves specifying the quantum circuits that realize the key operators: the phase-flip operators $S_f$ and $S_{\vert 0 \rangle ^ {\otimes n}}$ in the QAAA and the phase-rotation operators $R_{f}^{\phi}$ and $R_{\vert 0 \rangle ^ {\otimes n}} ^{\phi}$ in the EQAAA.

In the QAAA, the operator $S_f$, a key component of the amplitude amplification operator $Q$, flips the phase of a target basis state $\vert x^{(k)}  \rangle = \vert  x _{0}^{(k)}x_{1}^{(k)}\cdots x_{n-1}^{(k)}\rangle$, where $x_{i}^{(k)}\in\{0,1\}$ and $k\in\{0,1,\cdots, \vert X_g\vert-1\}$. The circuit for $S_f$ (Figure \ref{operator-Sf}) has a depth of 3 and is constructed as follows:
\begin{eqnarray}
S_f = \left(\bigotimes_{i=0} ^{n-1} X^{1-x^{(k)}_i}\right) \left(C^{n-1}Z\right) \left(\bigotimes_{i=0} ^{n-1} X^{1-x^{(k)}_i}\right).
\label{eq-Sf}
\end{eqnarray}
The multi-controlled $Z$ gate, $C^{n-1} Z$, flips the phase only of the $\vert 1\rangle ^{\otimes n}$ state. The surrounding $X$ gates, conditioned via the exponent $1-x^{(k)}_i$, effectively map the specific target basis state $\vert x^{(k)}  \rangle$ to $\vert 1\rangle ^{\otimes n}$ before the phase flip, and then map it back.

For a system with $\vert X_g\vert$ target basis states, the complete quantum circuit is constructed by sequentially applying the respective $S_f$ subcircuit for each target state. Each $S_f$ subcircuit, which marks a specific target basis state $\vert x^{(k)}  \rangle$  has a base depth of 3. Since the subcircuits are applied in series, the total circuit depth is the sum of their individual depths, yielding $3\vert X_g\vert$. This value represents the depth before any circuit-level optimizations are applied.

A notable case is when the target is the all-zero state $\vert 0\rangle ^{\otimes n}$. The corresponding circuit, which has a depth of 3, simplifies to
\begin{eqnarray}
S_{\vert 0 \rangle ^ {\otimes n}} = \left( X ^{\otimes n} \right) \left(C^{n-1}Z\right) \left( X ^{\otimes n} \right),
\label{eq-S0}
\end{eqnarray}
as depicted in Figure \ref{operator-S0}.

\begin{figure}[H]
    \centering
    \subfloat[$S_f$]
    {
        \includegraphics[width=0.4\textwidth]{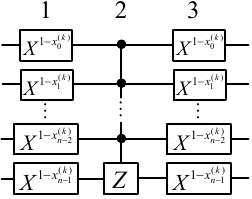}
        \label{operator-Sf}
    }
    \hspace{5pt}
    \subfloat[$S_{\vert 0 \rangle ^ {\otimes n}}$]
    {
        \includegraphics[width=0.35\textwidth]{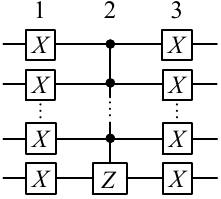}
        \label{operator-S0}
    }
    \caption{The quantum circuits corresponding to $S_f$ and $S_{\vert 0 \rangle ^ {\otimes n}}$ operators.}
    \label{operatorSfS0}
\end{figure}

The phase rotation operator $R_{f}^{\phi}$ in the EQAAA shares a similar construction principle with the $S_f$ in QAAA. Its quantum circuit implementation, depicted in Figure \ref{operator-Rf}, is formally defined for a general target basis state $\vert x^{(k)}  \rangle$ as:
\begin{eqnarray}
R_{f}^{\phi} = \left(\bigotimes_{i=0} ^{n-1} X^{1-x^{(k)}_i}\right)\left( C^{n-1}PS(\phi)\right) \left(\bigotimes_{i=0} ^{n-1} X^{1-x^{(k)}_i}\right).
\label{eq-Rf}
\end{eqnarray}

\begin{figure}[H]
    \centering
    \subfloat[$R_{f}^{\phi}$]
    {
        \includegraphics[width=0.4\textwidth]{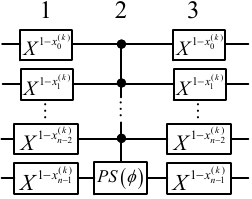}
        \label{operator-Rf}
    }
    \hspace{5pt}
    \subfloat[$R_{\vert 0 \rangle ^ {\otimes n}} ^{\phi}$]
    {
        \includegraphics[width=0.35\textwidth]{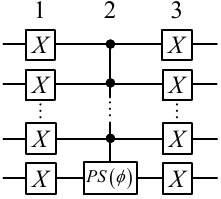}
        \label{operator-R0}
    }
    \caption{The quantum circuits corresponding to $R_{f}^{\phi}$ and $R_{\vert 0 \rangle ^ {\otimes n}} ^{\phi}$  operators.}
    \label{operatorRfR0}
\end{figure}

Here, the core phase operation is enacted by the multi-controlled phase shift gate $C^{n-1}PS(\phi)$, where the single-qubit phase gate is
\begin{equation}
 PS(\phi) = \left[
  \begin{array}{cc}
    1 & 0 \\
    0 & e^{i\phi}
  \end{array}
\right].
\end{equation}
The surrounding X gates perform the same basis transformation role.  Specifically, for the all-zero state $\vert 0\rangle ^{\otimes n}$, the construction (see Figure \ref{operator-R0}) is
\begin{eqnarray}
R_{\vert 0 \rangle ^ {\otimes n}} ^{\phi} = \left( X ^{\otimes n} \right) \left(C^{n-1}PS(\phi)\right) \left( X ^{\otimes n} \right).
\label{eq-R0}
\end{eqnarray}
These operators apply a conditional phase shift of $e^{i\phi}$ specifically to their respective target basis states, leaving other basis states unchanged.

As a supplementary note, the circuit depth of the $C^{n-1}PS(\phi)$ gate is 1. For a realistic assessment of the overall quantum circuit depth, a further decomposition of this multi-qubit gate into a sequence of native single- and two-qubit gates is necessary. One specific decomposition method will be presented in Section \ref{sec-experiment} and Appendix \ref{appx-decom}.

Next, we give the circuit depths of QAAA and EQAAA by the following theorems, respectively.

\begin{theorem}
(\textbf{Circuit Depth of QAAA}) The circuit depth of Algorithm \ref{algo-QAA} is given by
\begin{eqnarray}
\text{dep}(\text{QAAA}, p_g, \vert X_g\vert) 
= (2r + 1) \cdot \text{dep}(\mathcal{A}) + r \cdot (3\vert X_g\vert + 3),
\label{depthqaaa}
\end{eqnarray}
where
\begin{eqnarray}
    r = \left\lfloor \frac{\pi}{4\arcsin\left(\sqrt{p_g}\right)} \right\rfloor,
\end{eqnarray}
with $p_g$ denoting the probability of measuring a target string within the global target set $X_g$, and $\vert X_g\vert$ representing the number of target strings contained in $X_g$.
\label{theorem2qaaadepth}
\end{theorem}

\begin{proof}
The detailed proof is shown in Appendix \ref{prooftheorem2}. 
\end{proof}

\begin{theorem}
(\textbf{Circuit Depth of EQAAA}) The circuit depth of Algorithm \ref{algo-EQAA} is given by 
\begin{eqnarray}
\text{dep}(\text{EQAAA}, p_g, \vert X_g\vert)
=  \left( 2J +3 \right) \cdot \text{dep}(\mathcal{A}) + \left( J+1 \right) \cdot \left(3\vert X_g\vert+3\right),
\label{deptheqaaa}
\end{eqnarray}
where 
\begin{eqnarray}
J = \left\lfloor \frac{\pi}{4\arcsin\left(\sqrt{p_g}\right)} - \frac{1}{2} \right\rfloor,
\end{eqnarray}
with $p_g$ denoting the probability of measuring a target string within the global target set $X_g$, and $\vert X_g\vert$ representing the number of target strings contained in $X_g$.
\label{theorem3eqaaadepth}
\end{theorem}

\begin{proof}
The detailed proof is shown in Appendix \ref{prooftheorem3}. 
\end{proof}

It follows from the above two theorems that the circuit depths of the QAAA and the EQAAA depend on three factors, namely $\text{dep}(\mathcal{A})$, $p_g$, and $\vert X_g\vert$.

Next, the circuit depth of the DEQAAA is formally characterized by the following theorem.

\begin{theorem}\label{theorem4deqaaadepth}
(\textbf{Circuit Depth of DEQAAA}) The circuit depth of Algorithm \ref{algo-DEQAA} is given by 
\begin{eqnarray}
\text{dep}(\text{DEQAAA}) = 
\begin{cases}
\text{dep}(\text{First phase}) , & p'_g=1 , \\
\text{dep}(\text{First phase}) + \text{dep}(\text{Second phase}) , & \text{otherwise},
\end{cases}
\label{depthdeqaaa}
\end{eqnarray}
where
\begin{eqnarray}
\text{dep}(\text{First phase}) &=& \text{dep}(\mathcal{A}) + \max\limits_{j \in \{0,1,\cdots,t-1 \}} \left\{ \text{dep}(\text{$j$-th node, $p_j$, $\vert X_j\vert$})\right\},\\
\text{dep}(\text{$j$-th node, $p_j$, $\vert X_j\vert$}) &=& \left( J_j +1 \right) \cdot \left(3\vert X_j\vert + 2\text{dep}(\mathcal{A}_{\vert \varphi_j \rangle}) + 3\right),\\
J_j &=& \left\lfloor \frac{\pi}{4\arcsin\left(\sqrt{p_j}\right)} - \frac{1}{2} \right\rfloor, \\
\text{dep}(\text{Second phase}) &=&\left( \hat{J} +1 \right) \cdot \left(3\vert X_g\vert + 2\text{dep}(\text{First phase}) + 3\right),\\
 \hat{J}&=& \left\lfloor \frac{\pi}{4\arcsin\left(\sqrt{p'_g}\right)} - \frac{1}{2} \right\rfloor, 
\end{eqnarray}
with $p'_g$ denoting the success probability of the second phase, $p_j$ and $\vert X_j\vert$ representing the success probability of the $j$-th node and the number of target strings in the local target set $X_j$, respectively, and $\mathcal{A}_{\vert \varphi_j \rangle}$ being the local unitary preparation operator of the $j$-th node (generating substate $\vert\varphi_j\rangle = \mathcal{A}_{\vert \varphi_j \rangle}\vert 0\rangle^{\otimes n_j}$).
\end{theorem}

\begin{proof}
The detailed proof is shown in Appendix \ref{prooftheorem4}. 
\end{proof}

We have now successfully derived the circuit depth of the DEQAAA. For fixed $p_g$ and $\vert X_g\vert$, the partitioning strategy plays a pivotal role: it first directly determines the number of nodes $t$ and the number of qubits $n_j$ on the $j$-th node, thereby influencing the success probability $p_j$ of the $j$-th node and the number of target strings $\vert X_j\vert$ in the local target set $X_j$, and these two are the core factors that determine the $\text{dep}(\text{DEQAAA})$. 

In other words, based on $p_g$, $\vert X_g\vert$, and the exact probability distribution $P$ of $\vert \Psi\rangle$, selecting an appropriate partitioning strategy enables $\max_{j \in \{0,1,\cdots,t-1 \}} \left\{ \text{dep}(j\text{-th node}, p_j, \vert X_j\vert)\right\}$ to reach its minimum value, thereby achieving the optimal circuit depth of the DEQAAA. We do not elaborate further on this.

Additionally, two issues remain to be further explored in future research: first, whether $p'_g$ is sufficiently improved compared to $p_g$ after the first phase; second, beyond what threshold the decomposed circuit depth of the $C^{n-1}PS(\phi)$ gate makes $\text{dep}(\text{DEQAAA})$ significantly advantageous over those of $\text{dep}(\text{QAAA})$ and $\text{dep}(\text{EQAAA})$.

\subsection{Comparison}\label{comparison}
Within this subsection, we present a comparative analysis of DEQAAA with existing QAAAs (QAAA \cite{brassard2002quantum}, EQAAA, DQAAA \cite{hua2025distributed}), focusing on key performance metrics including circuit depth, exact solution capability, scalability, and resource requirements, which are all critical for assessing the practicality of quantum algorithms. 

DEQAAA shares core consistent features with three existing QAAAs: all support multi-target amplitude amplification, require no auxiliary qubits, and can be implemented via quantum simulation software, with consistent basic functions and implementation convenience. In addition, DEQAAA, QAAA, and EQAAA all have no stringent assumption constraints on $\mathcal{A}$, and DEQAAA belongs to the same distributed architecture category as DQAAA, with commonalities in algorithm design logic and architecture selection, providing a clear benchmark for subsequent targeted comparisons.

The key differences between DEQAAA and two single-node algorithms (QAAA, EQAAA) focus on solution accuracy and architectural scalability. QAAA only supports approximate amplitude amplification, adopts a single-node architecture (a single node hosts $n$ qubits, circuit depth see Eq.~\eqref{depthqaaa}), and as the number of qubits increases, the single-node load surges, resulting in limited scalability in large-scale scenarios. Although EQAAA supports multi-target exact solutions, it retains the single-node design (circuit depth see Eq.~\eqref{deptheqaaa}) and faces the same single-node load and scalability bottlenecks as QAAA. DEQAAA breaks through the above limitations: it not only achieves exact solution capability but also innovatively adopts a distributed architecture (flexibly configurable nodes $2 \leq t \leq n$), with the maximum qubits per node optimized to $\max\left(n_0,n_1,\cdots,n_{t-1}\right)$. The total number of qubits remains $n$, and the circuit depth can be optimized (Eq.~\eqref{depthdeqaaa}), greatly improving adaptability in large-scale scenarios while retaining core functional advantages.

Compared with DQAAA, another representative algorithm of the distributed architecture, DEQAAA achieves pivotal upgrades in scenario versatility, resource utilization efficiency, and solution accuracy. DQAAA imposes stringent assumptions on $\mathcal{A}$, only adapts to initial states with specific amplitude distributions, has limited applicable scenarios, and its number of nodes is fixed as a power of 2 ($2^j$, $1 \leq j \leq n-1$), which cannot be compatible with other node configurations. Its total number of qubits reaches $2^j(n-j)$, with significant resource redundancy, and only approximate solutions can be completed. In contrast, DEQAAA lifts the restrictive assumptions on $\mathcal{A}$, enabling compatibility with arbitrary initial states and flexible node configurations in the range $2 \leq t \leq n$. Critically, it maintains a fixed total qubit count of $n$, which markedly enhances resource utilization efficiency. Furthermore, DEQAAA supports exact amplitude amplification, and its circuit depth can be further optimized via qubit partitioning strategies, thereby yielding significant improvements in comprehensive performance and scenario adaptability over DQAAA.

In summary, through distributed architecture optimization, DEQAAA retains core functional advantages while solving the scalability bottlenecks of single-node algorithms (e.g., QAAA, EQAAA) and overcoming the issues of insufficient versatility and resource redundancy of distributed algorithms (e.g., DQAAA), achieving better comprehensive performance. For a more straightforward benchmarking of different methods, the following table is provided (see Table \ref{algorithmscompare}).

\begin{table}[H]
  \centering
  \caption{Comparison between DEQAAA and existing QAAAs.}
\scalebox{0.8}{
  \begin{tabular}{lccccc}
    \toprule
Performance Metrics & QAAA \cite{brassard2002quantum} & EQAAA & DQAAA \cite{hua2025distributed} & \textbf{DEQAAA}  \\
    \midrule
\text{1. Amplitude amplification for multiple targets} &\text{Yes} &\text{Yes}&\text{Yes}&\textbf{Yes}\\
\text{2. Exact solution achievement}
&\text{No}&\text{Yes}&\text{No}&\textbf{Yes} \\
\text{3. Stringent assumptions for $\mathcal{A}$}
&\text{No}&\text{No}&\text{Yes}&\textbf{No} \\
\text{4. Arbitrary amplitude distribution of initial state $\vert \Psi\rangle$}
&\text{Yes}&\text{Yes}&\text{No}&\textbf{Yes} \\
\text{5. Number of computing nodes} &$1$&$1$& $2^j$ ($1 \leq j \leq n-1$)&$\bm{2 \leq t \leq n}$ \\
\text{6. Maximum qubits at a single node} &$n$ &$n$&$n-j$&$\bm{ \max\left(n_0,n_1,\cdots,n_{t-1}\right)}$\\
\text{7. Total number of qubits}&$n$&$n$&$2^j(n-j)$&$\bm{n}$\\
\text{8. Requirement for auxiliary qubits}& \text{No}&\text{No}&\text{No}&\textbf{No}\\
\text{9. Implementation via quantum simulation software}
&\text{Yes}&\text{Yes}&\text{Yes}&\textbf{Yes} \\
\text{10. Quantum circuit depth}
&Eq.~\eqref{depthqaaa}&Eq.~\eqref{deptheqaaa}&-&\textbf{Eq.~\eqref{depthdeqaaa} }\\
    \bottomrule
  \end{tabular}}
\label{algorithmscompare}
\end{table}

\section{Experiment}\label{sec-experiment}
This section elaborates on the workflows and simulation results of QAAA, EQAAA, and DEQAAA, respectively, and verifies the correctness and effectiveness of DEQAAA through the experimental results. All experiments were implemented based on the MindSpore Quantum framework \cite{mq_2021, xu2024mindspore}.

A preliminary note is warranted: the readout orders of qubits in theory and experiments are reversed due to the default little-endian read/write mode of MindSpore Quantum. Specifically, the theoretical qubit $q_i$ corresponds to the experimental qubit $q_{n-1-i}$. It is emphasized that this discrepancy neither affects the computational processes or results of the aforementioned formulas nor hinders the comprehension of the case study implementation. The only implication is that the node positions adopted in subsequent experiments are reversed relative to the theoretical configuration.

\subsection{The Probability Distribution}
To begin with, assume an operator $\mathcal{A}$ (shown in Figure \ref{fig-example-A-circuit}) prepares the 4-qubit state $\vert\Psi\rangle = \mathcal{A}\vert0\rangle^{\otimes 4}$, with $\vert\Psi\rangle$ expressed as:
\begin{eqnarray}
\vert\Psi\rangle&=&\mathcal{A}\vert 0\rangle^{\otimes 4}\\
&=&0.1506\vert0000\rangle+0.1908\vert0001\rangle+0.3120\vert0010\rangle+0.1788\vert0011\rangle\nonumber \\
&+&0.2055\vert0100\rangle+0.2719\vert0101\rangle+0.2793\vert0110\rangle+0.2273\vert0111\rangle\nonumber \\
&+&0.3164\vert1000\rangle+0.2719\vert1001\rangle+0.3180\vert1010\rangle+0.2207\vert1011\rangle\nonumber \\
&+&0.1860\vert1100\rangle+0.2572\vert1101\rangle+0.3046\vert1110\rangle+0.2200\vert1111\rangle.
\label{eq-exact-distribution}
\end{eqnarray}

\begin{figure}[H]
\centering
\includegraphics[width=0.7\textwidth]{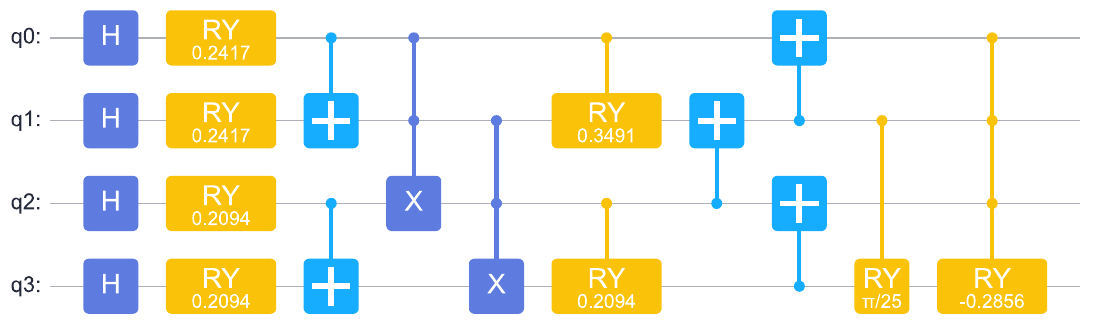}
\caption{Quantum circuit of the operator $\mathcal{A}$ to prepare the 4-qubit state $\vert\Psi\rangle$.}
\label{fig-example-A-circuit}
\end{figure}

Suppose the global target set is defined as $X_g=\{8,14\}=\{1000,1110\}$. Based on the amplitude distribution of $\vert\Psi\rangle$, the target measurement probability (success probability) is:
\begin{equation}
    p_g=0.3164^2+0.3046^2=0.1929.
    \label{eq-example-pg}
\end{equation}

It is well-established that the square of the amplitude of a quantum state corresponds to the measurement probability, so if the amplitude distribution of $\vert\Psi\rangle$ can be precisely acquired, the exact probability distribution $P$ can be derived by squaring the amplitudes. However, this is hardly achievable in practical physical experiments. In practice, repeated measurements on $\vert\Psi\rangle$ yield the approximate probability distribution $\tilde{P}$, and the more measurements are conducted, the closer $\tilde{P}$ converges to $P$. To quantify the impact of measurement count on precision, we use the Kullback–Leibler (KL) divergence \cite{Kullback1951OnIA}, defined as:
\begin{equation}
D_{KL}(\tilde{P}\vert\vert P)=\sum_{x}\tilde{P}(x)\log\frac{\tilde{P}(x)}{P(x)}.
\label{eq-KLdivergence}
\end{equation}
A smaller $D_{KL}(\tilde{P}\vert\vert P)$ (closer to 0) indicates better consistency between $\tilde{P}$ and $P$.

The simulator enables repeated measurements of $\vert\Psi\rangle$ to generate $\tilde{P}$. The distribution plots for 10,000 and 100,000 measurements are presented in Figure \ref{fig-example-A-result-1w} and Figure \ref{fig-example-A-result-10w}, respectively.
\begin{figure}[H]
    \centering
    \subfloat[10,000 times]
    {
        \includegraphics[width=0.4\textwidth]{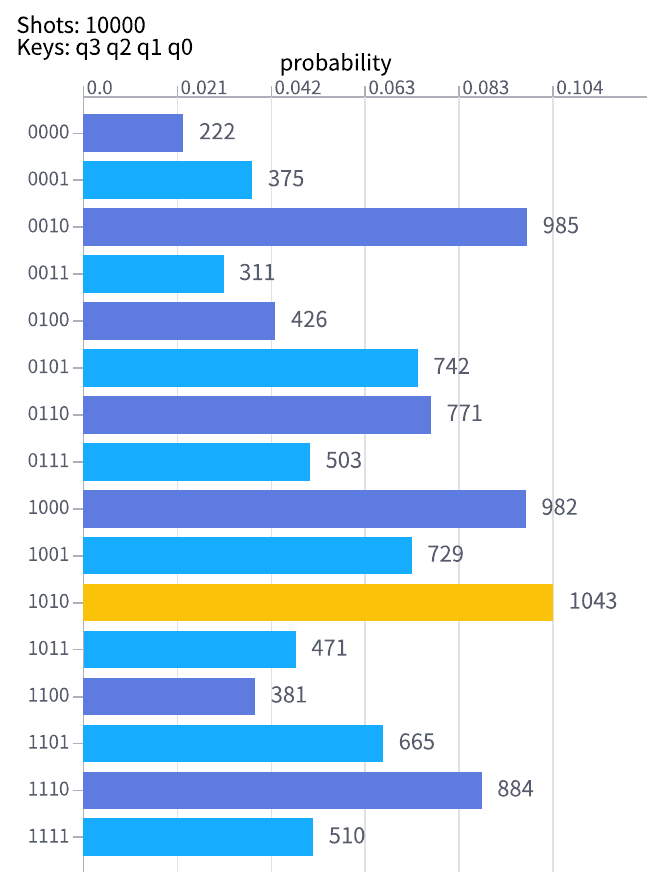}
        \label{fig-example-A-result-1w}
    }
    \hspace{5pt}
    \subfloat[100,000 times]
    {
        \includegraphics[width=0.4\textwidth]{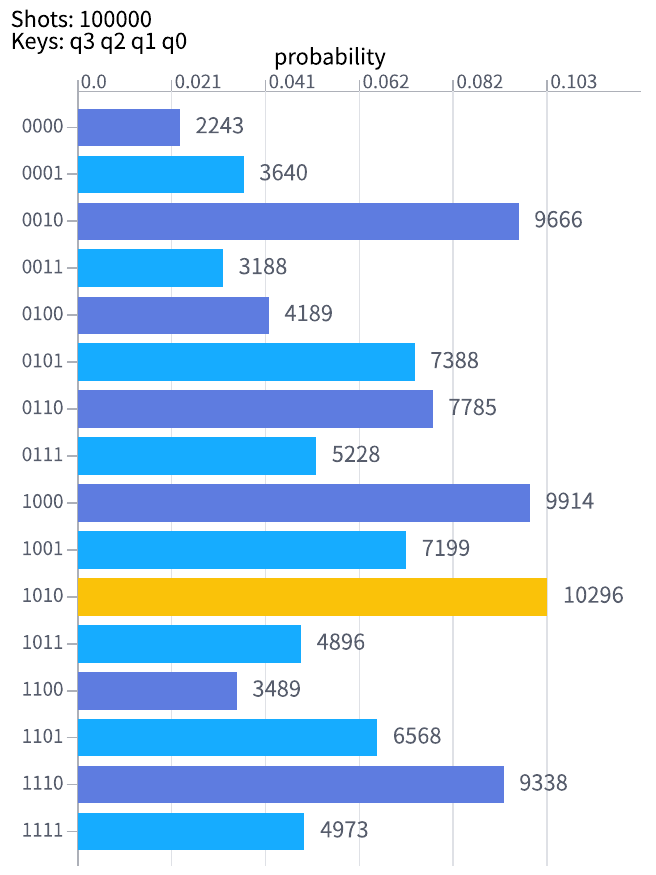}
        \label{fig-example-A-result-10w}
    }
    \caption{The distribution plots of measurement outcomes of $\vert\Psi\rangle$.}
    \label{fig-example-A-result}
\end{figure}

It should be briefly noted that the simulator supports random seed configuration. Different seeds generally yield distinct probability distributions, while the same seed ensures the reproducibility of identical distributions. Here, the random seed is uniformly set to 21 for consistency.

As calculated, the KL divergences between the approximate distributions (for 10,000 and 100,000 measurements) and the exact probability distribution $P$ (derived from squaring the amplitudes of $\vert\Psi\rangle$ in Eq.~\eqref{eq-exact-distribution}) are as follows:
\begin{eqnarray}
    D_{KL}^{(10,000)}&=5.1208\times 10^{-4},\\
    D_{KL}^{(100,000)}&=7.9374\times 10^{-5}.
\end{eqnarray}
As demonstrated, the latter (for 100,000 measurements) aligns more closely with the theoretical probability distribution. Hence, obtaining an acceptable approximation of the probability distribution hinges on repeated measurements, a requirement that is circumvented in this study by the premise that the distribution can be determined exactly.

\subsection{The 4-qubit QAAA with $X_g=\{8, 14\}$}
According to Algorithm \ref{algo-QAA}, the quantum circuit for this case can be constructed, as shown in Figure \ref{fig-example-QAA-circuit}. Specifically, the circuits described in Eq.~\eqref{eq-Sf} and Eq.~\eqref{eq-S0} are employed to build the phase-flip operators. 

\begin{figure}[H]
\centering
\includegraphics[width=\textwidth]{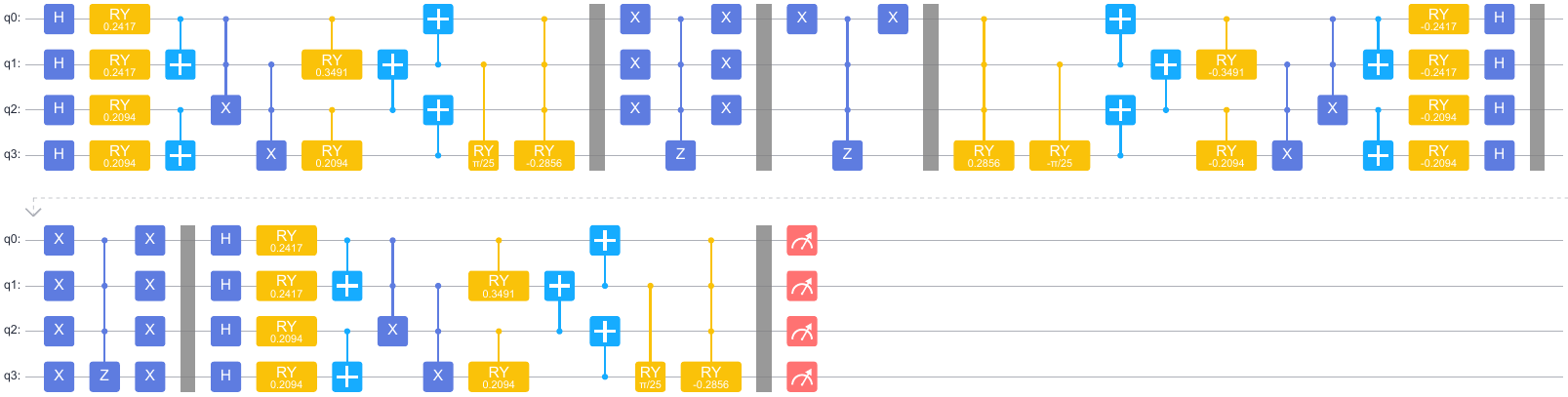}
\vspace{-1.8em}
\caption{Quantum circuit of the 4-qubit QAAA with the set of target strings $X_g=\{8, 14\}$. The number of repetitions is $r=\left\lfloor\pi /\left(4\arcsin\left(\sqrt{p_g}\right)\right)\right\rfloor=1$. The number of gates is 76 and the circuit depth is 39.}
\label{fig-example-QAA-circuit}
\end{figure}

By measuring the circuit in Figure \ref{fig-example-QAA-circuit} 10,000 times, the measurement results are presented in Figure \ref{fig-example-QAA-result}. Specifically, QAAA amplifies the amplitudes of the target states $\vert 1000\rangle$ and $\vert 1110\rangle$, increasing the success probability from the initial $0.1929$ to $0.9595$.

\begin{figure}[H]
\centering
\includegraphics[width=0.37\textwidth]{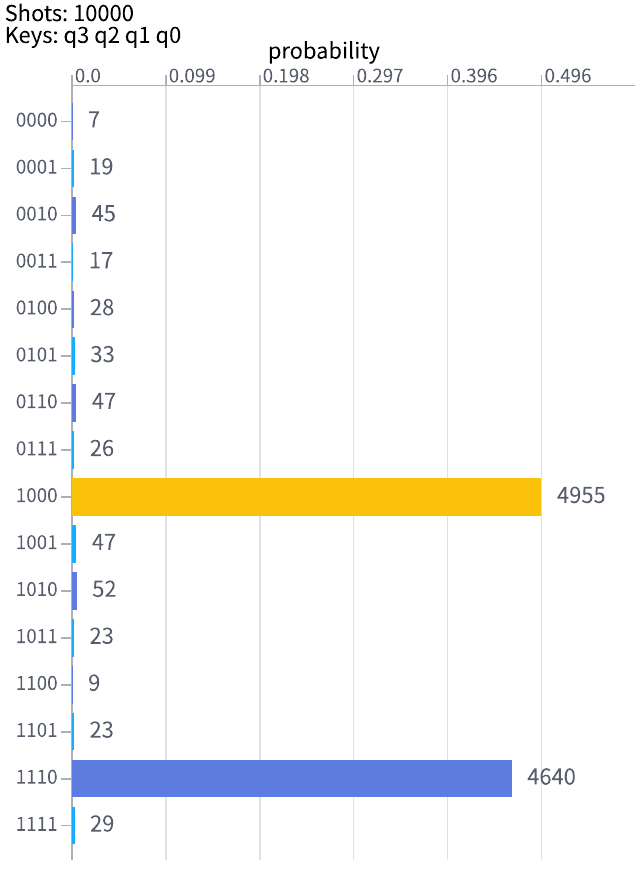}
\caption{Measurement distribution of $\vert\Psi\rangle$ after QAAA (10,000 measurements).}
\label{fig-example-QAA-result}
\end{figure}

\subsection{The 4-qubit EQAAA with $X_g=\{8, 14\}$}
We then employ EQAAA to enhance the success probability to 1. Similarly, the complete quantum circuit for this case is constructed via Algorithm \ref{algo-EQAA}, as shown in Figure \ref{fig-example-EQAA-circuit}. Specifically, the phase rotation operators are built following the circuits described in Eq.~\eqref{eq-Rf} and Eq.~\eqref{eq-R0}.

\begin{figure}[H]
\centering
\includegraphics[width=\textwidth]{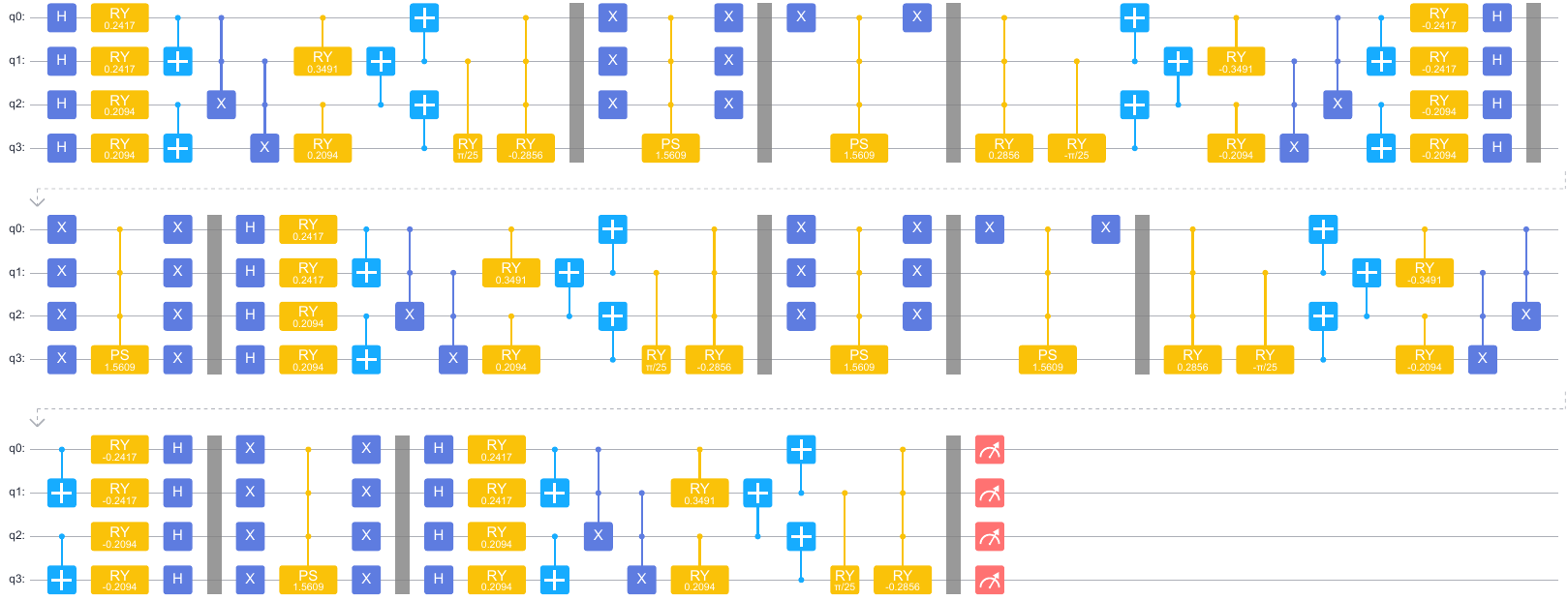}
\vspace{-1.8em}
\caption{Quantum circuit of the 4-qubit EQAAA with the set of target strings $X_g=\{8, 14\}$. The number of repetitions is $J+1=\left\lfloor \pi/\left(4\arcsin\left(\sqrt{p_g}\right)\right) - 1/2 \right\rfloor+1=2$, and the phase angle is $\phi = 2 \arcsin \left(\frac{\sin \left(\frac{\pi}{4 J + 6}\right)}{\sqrt{p_g}}\right)=1.5609$. The number of gates is 133 and the circuit depth is 68.}
\label{fig-example-EQAA-circuit}
\end{figure}

By measuring the circuit in Figure \ref{fig-example-EQAA-circuit} 10,000 times, the measurement results are presented in Figure \ref{fig-example-EQAA-result}. Specifically, EQAAA achieves exact amplitude amplification of the target states $\vert 1000\rangle$ and $\vert 1110\rangle$, realizing a success probability of 1.

\begin{figure}[H]
\centering
\includegraphics[width=0.45\textwidth]{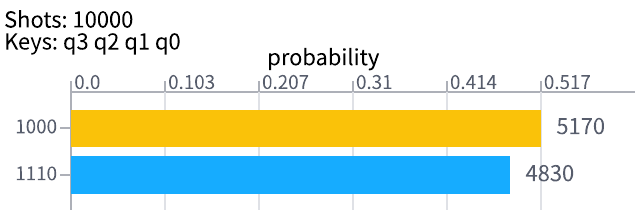}
\caption{Measurement distribution of $\vert\Psi\rangle$ after EQAAA (10,000 measurements).}
\label{fig-example-EQAA-result}
\end{figure}

\subsection{The 4-qubit DEQAAA with $t=2$ nodes and $X_g=\{8, 14\}$}
Lastly, we extend our exploration to the application of DEQAAA for this example. Specifically, we assume a small-scale distributed system consisting of $t=2$ quantum computing nodes, where each node is equipped with $n_0 = n_1 = 2$ qubits, resulting in a total of $n=4$ computing qubits.

First, following the procedure outlined in Algorithm \ref{algo-substate}, the substate for each node is computed as follows:
\begin{eqnarray}
    \vert\varphi_0\rangle&=&0.4340\vert 00\rangle+0.4958\vert 01\rangle+0.5691\vert 10\rangle+0.4919\vert 11\rangle,\\
    \vert\varphi_1\rangle&=&0.4468\vert 00\rangle+0.5004\vert 01\rangle+0.6077\vert 10\rangle+0.4251\vert 11\rangle.
\end{eqnarray}
According to Eq.~\eqref{eq-target-set}, the local target sets for each node are directly derived from $X_g=\{8, 14\}=\{1000,1110\}$:
\begin{equation}
    X_0=\{10,11\},\quad X_1=\{00,10\}.
\end{equation}
Further, according to Eq.~\eqref{eq-pj}, the probabilities for each node are calculated as the sum of the probabilities of measuring each basis state within their respective local target sets:
\begin{eqnarray}
    p_0&=&\vert0.5691\vert^2+\vert0.4919\vert^2=0.5658,\\
    p_1&=&\vert0.4468\vert^2+\vert0.6077\vert^2=0.5689.
\end{eqnarray}

Subsequently, local EQAAA is executed on each node independently, resulting in the quantum circuit for the first phase as illustrated in Figure \ref{fig-example-DEQAA-stage1-circuit}. It is worth noting that the operators $\mathcal{A}_{\vert\varphi_0\rangle}$ and $\mathcal{A}_{\vert\varphi_1\rangle}$ employed for substate generation are not unique. In this work, we opt to adopt Top-Down amplitude encoding circuits \cite{araujo2023configurable} to realize these operators.

\begin{figure}[H]
\centering
\includegraphics[width=\textwidth]{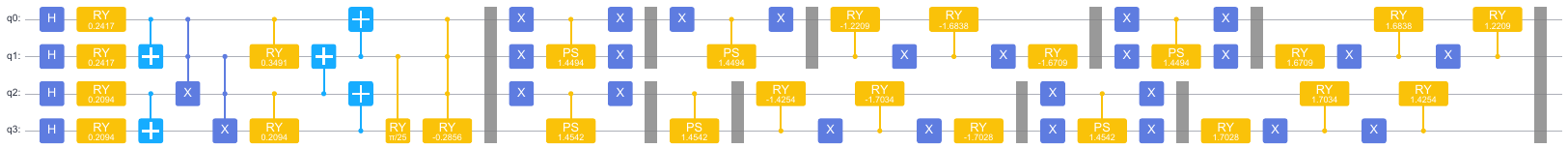}
\vspace{-1.8em}
\caption{Quantum circuit of 4-qubit DEQAAA (first phase) with $t=2$ nodes and target set $X_g=\{8, 14\}$. The repetition numbers are $J_0+1= \left\lfloor \frac{\pi}{4\arcsin\left(\sqrt{p_0}\right)} - \frac{1}{2} \right\rfloor+1=1$ and $J_1+1= \left\lfloor \frac{\pi}{4\arcsin\left(\sqrt{p_1}\right)} - \frac{1}{2} \right\rfloor+1=1$, and the phase angles are $\phi_0 = 2 \arcsin \left(\frac{\sin \left(\frac{\pi}{4 J_0 + 6}\right)}{\sqrt{p_0}}\right)=1.4542$ and $\phi_1 = 2 \arcsin \left(\frac{\sin \left(\frac{\pi}{4 J_1 + 6}\right)}{\sqrt{p_1}}\right)=1.4494$, respectively.}
\label{fig-example-DEQAA-stage1-circuit}
\end{figure}

By measuring the circuit in Figure \ref{fig-example-DEQAA-stage1-circuit} 10,000 times, the measurement results are presented in Figure \ref{fig-example-DEQAA-stage1-result}. It can be observed that the first phase of DEQAAA not only amplifies the amplitudes of $\vert 1000\rangle$ and $\vert 1110\rangle$ (the global target states) but also enhances those of $\vert1010\rangle$ and $\vert1100\rangle$.

\begin{figure}[H]
\centering
\includegraphics[width=0.36\textwidth]{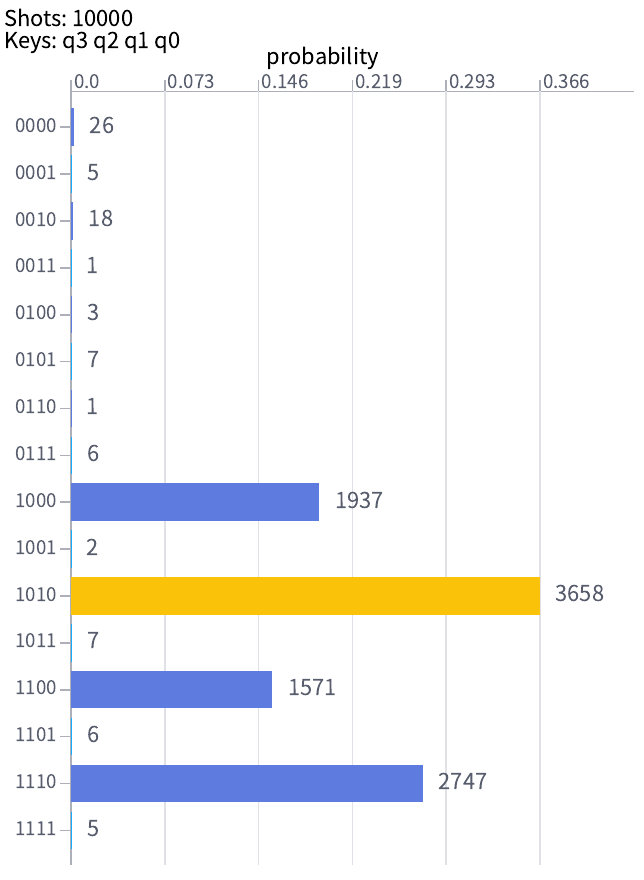}
\caption{Measurement distribution of $\vert\Psi\rangle$ after the first phase of DEQAAA (10,000 measurements).}
\label{fig-example-DEQAA-stage1-result}
\end{figure}

At this point, the circuit of the first phase is treated as a whole, corresponding to operator $\mathcal{B}$ in Figure \ref{fig-DEQAA}. The quantum state of the whole system then becomes
\begin{eqnarray}
\vert\Phi_1\rangle&=&\mathcal{B}\vert 0\rangle^{\otimes 4} \\
&=&(0.0106-0.0540j)\vert0000\rangle+(-0.0279-0.0107j)\vert0001\rangle +(-0.0051+0.0373j)\vert0010\rangle \nonumber \\
&+&(-0.0066-0.0051j)\vert0011\rangle+(0.0172-0.0014j)\vert0100\rangle+(0.0234+0.0181j)\vert0101\rangle \nonumber \\
&+&(-0.0137+0.0029j)\vert0110\rangle+(0.0166+0.0200j)\vert0111\rangle+(-0.0840-0.4318j)\vert1000\rangle \nonumber \\
&+&(0.0065-0.0122j)\vert1001\rangle+(-0.0477-0.6079j)\vert1010\rangle+(0.0048-0.0153j)\vert1011\rangle \nonumber \\
&+&(-0.0049-0.3907j)\vert1100\rangle+(-0.0111+0.0149j)\vert1101\rangle+(-0.0936-0.5143j)\vert1110\rangle \nonumber \\
&+&(-0.0089+0.0187j)\vert1111\rangle,
\end{eqnarray}
and the new global success probability is calculated as:
\begin{equation}
    p_g'=\vert-0.0840-0.4318j\vert^2+\vert -0.0936-0.5143j\vert^2=0.4667,
\end{equation}
where $j$ denotes the imaginary unit.

Although the global success probability of the target strings has increased from the initial $0.1929$ to $0.4667$ after the first phase of DEQAAA, exact amplitude amplification of the target strings as required has not yet been achieved. Thus, it is necessary to proceed to the second phase, where EQAAA is applied to the entire distributed system, yielding the complete DEQAAA circuit illustrated in Figure \ref{fig-example-DEQAA-circuit}.

\begin{figure}[H]
\centering
\includegraphics[width=\textwidth]{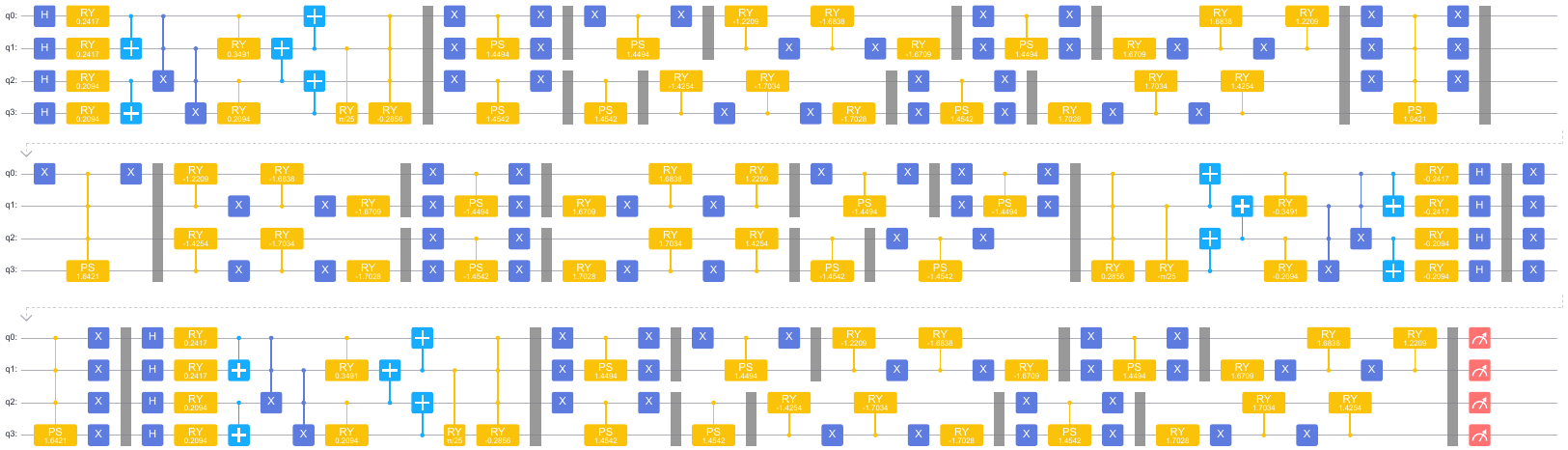}
\vspace{-1.8em}
\caption{Complete quantum circuit of the 4-qubit DEQAAA (first $\&$ second phases) with $t=2$ nodes and the set of target strings $X_g=\{8, 14\}$. In the second phase, the number of repetitions is $\hat{J}+1= \left\lfloor \frac{\pi}{4\arcsin\left(\sqrt{p'_g}\right)} - \frac{1}{2} \right\rfloor+1=1$, and the phase angle is $\hat{\phi} = 2 \arcsin \left(\frac{\sin \left(\frac{\pi}{4 \hat{J}+ 6}\right)}{\sqrt{p'_g}}\right)=1.6421$. The number of gates is 202 and the circuit depth is 96.}
\label{fig-example-DEQAA-circuit}
\end{figure}

By measuring the circuit in Figure \ref{fig-example-DEQAA-circuit} 10,000 times, the measurement results are presented in Figure \ref{fig-example-DEQAA-result}. It can be observed that the execution of the second phase successfully locks the target states to the global target set $X_g=\{1000,1110\}$, realizing exact amplitude amplification as intended.

\begin{figure}[H]
\centering
\includegraphics[width=0.45\textwidth]{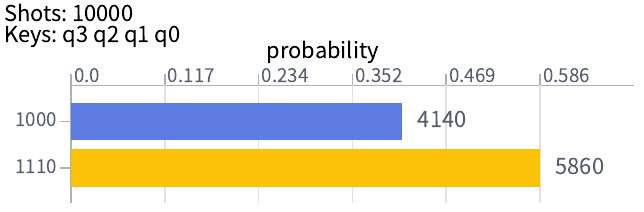}
\caption{Measurement distribution of $\vert\Psi\rangle$ after DEQAAA (10,000 measurements).}
\label{fig-example-DEQAA-result}
\end{figure}

\subsection{Data comparison}\label{sec-datacomparison}
For a comprehensive performance comparison of the three algorithms, we present a summary of the key experimental results across multiple dimensions in Table \ref{table-example-compare}. Note that the global target set $\{8, 14\}$ corresponds to $\{1000,1110\}$ in binary representation.

\begin{table}[H]
  \centering
  \caption{Comparison between 4-qubit QAAA, EQAAA and DEQAAA in this example.}
\scalebox{0.83}{
  \begin{tabular}{lccc}
    \toprule
Performance Metrics & QAAA & EQAAA & \textbf{DEQAAA} \\
    \midrule
\text{1. Number of computing nodes} & 1 & 1 & \textbf{2} \\

\text{2. Quantum circuit} & Figure \ref{fig-example-QAA-circuit} & Figure \ref{fig-example-EQAA-circuit} & \textbf{Figure \ref{fig-example-DEQAA-circuit}} \\

\text{3. Maximum qubits at a single node} & 4 & 4 & \textbf{2} \\

\text{4. Total number of qubits} & 4 & 4 & \textbf{4} \\

\text{5. Global target set $X_g$} & $\{8,14\}$ & $\{8,14\}$ & $\bm{\{8,14\}}$ \\

\text{6. Initial $p_g$ value} & 0.1929 & 0.1929 & \textbf{0.1929} \\

\text{7. Repetition number of amplification operators} & 1 & 2 & \textbf{1} \\

\text{8. Random seed value} & 21 & 21 & \textbf{21} \\

\text{9. Sampling results} & Figure \ref{fig-example-QAA-result} & Figure \ref{fig-example-EQAA-result} & \textbf{Figure \ref{fig-example-DEQAA-result}} \\

\text{10. Final success probability}
& 0.9595 & 1 & \textbf{1} \\

\text{11. Exact amplification} & No & Yes & \textbf{Yes} \\
    \bottomrule
  \end{tabular}}
\label{table-example-compare}
\end{table}

All experiments were conducted under unified conditions to guarantee comparative rigor. Three key conclusions are derived: (1) DEQAAA features a distributed architecture, reducing the maximum qubit demand per node to 2 while retaining a total of 4 qubits, thereby distinguishing itself from the single-node QAAA and EQAAA that each require 4 qubits. (2) DEQAAA accomplishes amplitude amplification with only 1 repetition of the amplification operator, matching the efficiency of QAAA and requiring 1 fewer repetition than EQAAA, which needs 2. (3) DEQAAA achieves exact amplitude amplification with a final success probability of 1, outperforming QAAA which only reaches 0.9595 and fails to realize exact amplification.

\begin{figure}[H]
    \centering
    \subfloat[4-Qubit]
    {
        \includegraphics[width=0.97\textwidth]{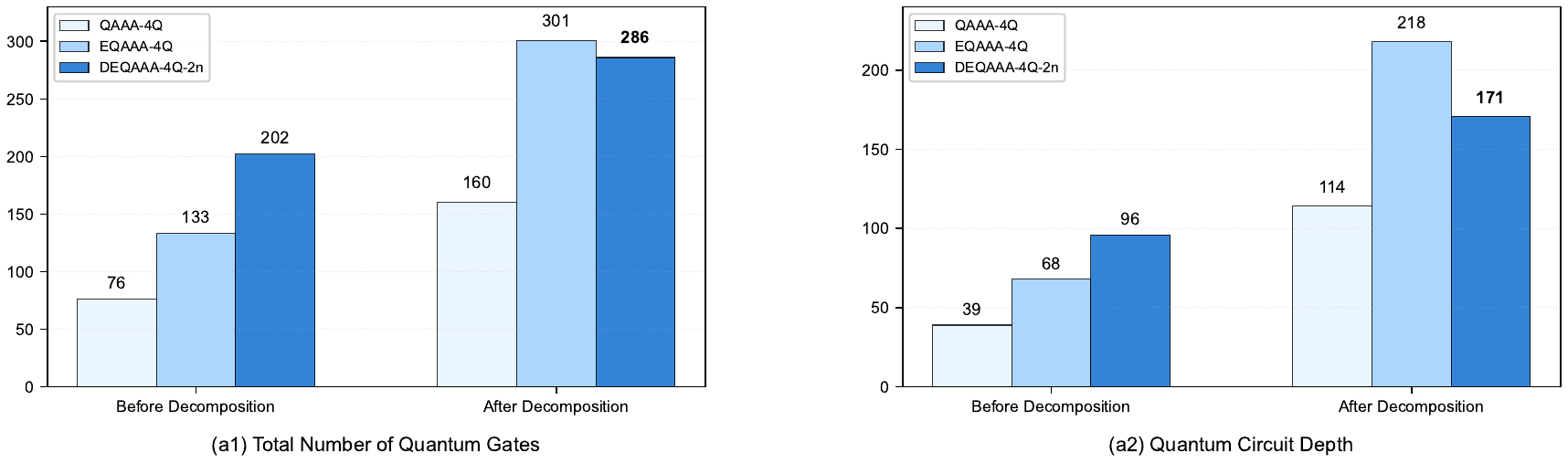}
        \label{fig-example-compare-4}
    }
    \vspace{1pt}
    \subfloat[6-Qubit]
    {
        \includegraphics[width=0.97\textwidth]{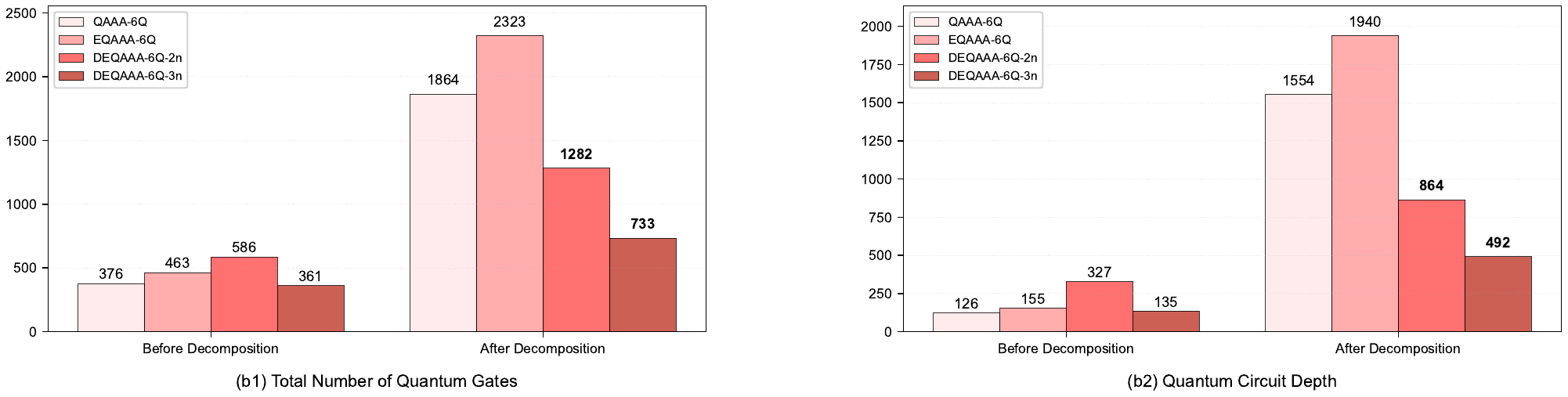}
        \label{fig-example-compare-6}
    }
    \vspace{1pt}
    \subfloat[8-Qubit]
    {
        \includegraphics[width=0.97\textwidth]{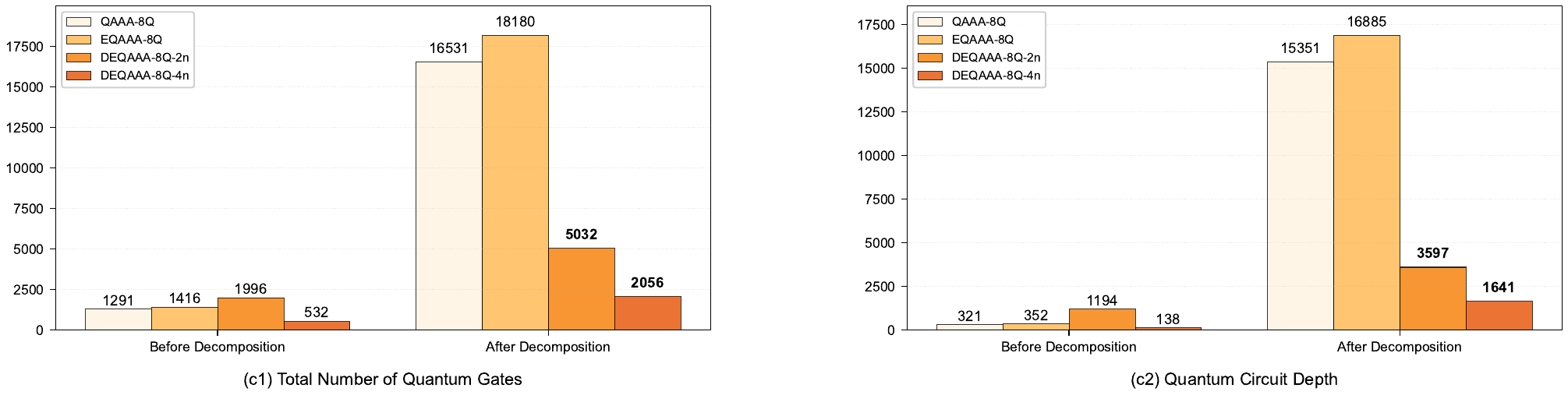}
        \label{fig-example-compare-8}
    }
    \vspace{1pt}
    \subfloat[10-Qubit]
    {
        \includegraphics[width=0.97\textwidth]{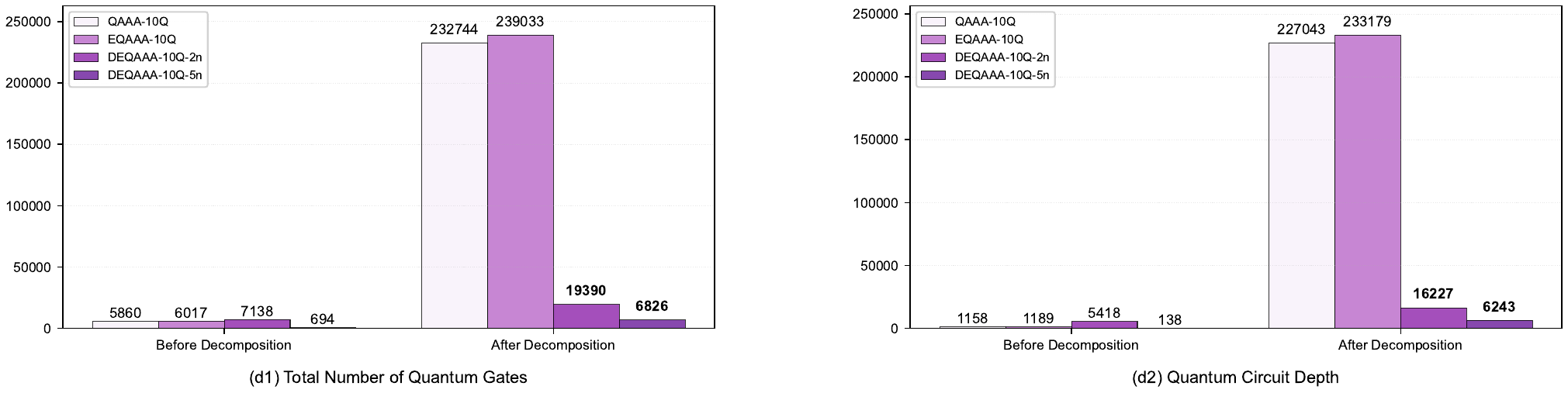}
        \label{fig-example-compare-10}
    }
    \caption{Comparison of the computational complexity of different QAAAs before and after decomposition.}
    \label{fig-example-compare}
\end{figure}

To further quantitatively demonstrate the intrinsic advantages of distributed quantum algorithms, we statistically compare the number of quantum gates and quantum circuit depth among the different algorithms. However, from the comparison in Figure \ref{fig-example-compare-4}, 4-qubit DEQAAA does not show advantages. We observe that the number of gates and circuit depth of 4-qubit DEQAAA are higher than those of the other two algorithms. This is because we count both 4-qubit and single-qubit gates as individual gates. In practice, a fair comparison requires decomposing multi-control gates into combinations of single- and two-qubit gates. Detailed information on the decomposition of $C^{3}Z$ and $C^{3}PS$ is provided in Appendix \ref{appx-decom} for reference.

After re-statistical analysis following decomposition, DEQAAA begins to show merits in both metrics, with 286 quantum gates and a circuit depth of 171. These merits stem from the reduced execution of multi-qubit gates, a key benefit brought by the distributed architecture of DEQAAA.

To validate the superiority of DEQAAA in terms of computational complexity, we extended our tests to 6-, 8-, and 10-qubit scenarios. The experimental parameters from these tests are summarized in Table \ref{table-example-compare-all-qubits}, with corresponding statistical results of quantum gate count and quantum circuit depth visualized in Figures \ref{fig-example-compare-6}, \ref{fig-example-compare-8}, and \ref{fig-example-compare-10}, respectively.

\begin{table}[H]
  \centering
  \caption{Comparison between QAAA, EQAAA and DEQAAA in 6-qubit, 8-qubit and 10-qubit examples.}
\scalebox{0.73}{
  \begin{tabular}{lccccccccc}
    \toprule
\multirow{2}{*}{Performance Metrics} & \multicolumn{3}{c}{6-qubit} & \multicolumn{3}{c}{8-qubit} & \multicolumn{3}{c}{10-qubit} \\
\cmidrule(lr){2-4} \cmidrule(lr){5-7} \cmidrule(lr){8-10}
& QAAA & EQAAA & \textbf{DEQAAA} & QAAA & EQAAA & \textbf{DEQAAA} & QAAA & EQAAA & \textbf{DEQAAA} \\
    \midrule
\text{1. Global target set $X_g$} & $\{8, 14\}$ & $\{8, 14\}$ & $\bm{\{8, 14\}}$ & $\{8, 14\}$ & $\{8, 14\}$ & $\bm{\{8, 14\}}$ & $\{8, 14\}$ & $\{8, 14\}$ & $\bm{\{8, 14\}}$ \\
\text{2. Initial $p_g$ value} & 0.0249 & 0.0249 & $\bm{0.0249}$ & 0.0054 & 0.0054 & $\bm{0.0054}$ & 0.0004 & 0.0004 & $\bm{0.0004}$ \\
\text{3. Final success probability} & 0.9799 & 1 & $\bm{1}$ & 0.9997 & 1 & $\bm{1}$ & 0.9999 & 1 & $\bm{1}$ \\
\text{4. Repetition number of amplification operators} & 4 & 5 & $\bm{1}$ & 10 & 11 & $\bm{1}$ & 37 & 38 & $\bm{1}$ \\
    \bottomrule
  \end{tabular}}
\label{table-example-compare-all-qubits}
\end{table}

Based on the simulation data from the 6-, 8-, and 10-qubit scenarios, DEQAAAs exhibit significant performance advantages and distinct characteristics after decomposition, which are mainly reflected in the following three aspects.

First, DEQAAAs show obvious advantages in reducing quantum gate count and circuit depth after decomposition when compared with EQAAA and QAAA. This superiority becomes more prominent as the number of qubits increases. As presented in Table  \ref{table-resource-overhead-compare}, the resource overhead reduction ratio of DEQAAA-Multi-n rises steadily with system expansion. It achieves over $97\%$ reduction in both gate count and depth in the 10-qubit scenario, a performance that proves its high efficiency in large-scale systems. 

Second, decomposition leads to increased quantum gate count and circuit depth for DEQAAAs, but the growth remains moderate, especially in large-qubit scenarios. For 10-qubit systems, DEQAAA's gate count and depth after decomposition are 9.86 and 45.24 times the original values, respectively. In contrast, QAAA and EQAAA see much sharper growth, with their 10-qubit gate count and depth reaching over 37.91 and 190.28 times the original values. This gap stems from DEQAAA's distributed architecture, which reduces dependence on multi-qubit gates and avoids excessive elementary gate expansion during decomposition.

\begin{table}[H]
  \centering
  \caption{Comparison of quantum gate count and circuit depth after decomposition (QAAA, EQAAA, DEQAAA)}
  \scalebox{0.73}{
  \begin{tabular}{lcccccccc}
    \toprule
    \multirow{3}{*}{Qubit Count} & \multirow{3}{*}{Metric} & \multicolumn{3}{c}{Algorithms} & \multicolumn{2}{c}{DEQAAA-Multi-n vs. QAAA} & \multicolumn{2}{c}{DEQAAA-Multi-n vs. EQAAA} \\
    \cmidrule(lr){3-5} \cmidrule(lr){6-7} \cmidrule(lr){8-9}
    & & QAAA & EQAAA & \textbf{DEQAAA-Multi-n} & Gate Reduction & Depth Reduction & Gate Reduction & Depth Reduction \\
    \midrule
    \multirow{2}{*}{4-Qubit} 
    & Total Gates & 160 & 301 & \textbf{286} & \multirow{2}{*}{-78.80\%} & \multirow{2}{*}{-50.00\%} & \multirow{2}{*}{5.00\%} & \multirow{2}{*}{21.60\%} \\
    & Circuit Depth & 114 & 218 & \textbf{171} &  &  &  &  \\
    \midrule
    \multirow{2}{*}{6-Qubit} 
    & Total Gates & 1864 & 2323 & \textbf{733} & \multirow{2}{*}{60.70\%} & \multirow{2}{*}{68.30\%} & \multirow{2}{*}{68.50\%} & \multirow{2}{*}{74.60\%} \\
    & Circuit Depth & 1554 & 1940 & \textbf{492} &  &  &  &  \\
    \midrule
    \multirow{2}{*}{8-Qubit} 
    & Total Gates & 16531 & 18180 & \textbf{2056} & \multirow{2}{*}{87.60\%} & \multirow{2}{*}{89.30\%} & \multirow{2}{*}{88.70\%} & \multirow{2}{*}{90.30\%} \\
    & Circuit Depth & 15351 & 16885 & \textbf{1641} &  &  &  &  \\
    \midrule
    \multirow{2}{*}{10-Qubit} 
    & Total Gates & 232744 & 239033 & \textbf{6826} & \multirow{2}{*}{\textbf{97.10\%}} & \multirow{2}{*}{\textbf{97.20\%}} & \multirow{2}{*}{\textbf{97.14\%}} & \multirow{2}{*}{\textbf{97.32\%}} \\
    & Circuit Depth & 227043 & 233179 & \textbf{6243} &  &  &  &  \\
    \bottomrule
  \end{tabular}}
  \label{table-resource-overhead-compare}
\end{table}

Third, DEQAAAs achieve continuous performance improvement through variant optimization with more qubits, showing superior scalability compared to EQAAA and QAAA. This is consistent with the trend in Figure \ref{fig-example-compare-trend} (logarithmic scale). Upgrading from 2-node (2n) to optimal multi-node configurations (3n for 6Q, 4n for 8Q, 5n for 10Q) significantly reduces gate count and depth. The optimization effect strengthens with system scale: gate count reduction rises from $42.90\%$ (6Q) to $64.80\%$ (10Q), and depth reduction increases from $43.10\%$ (6Q) to $61.50\%$ (10Q). This makes DEQAAA-Multi-n far more scalable in large-qubit scenarios.

\begin{figure}[H]
\centering
\includegraphics[width=\textwidth]{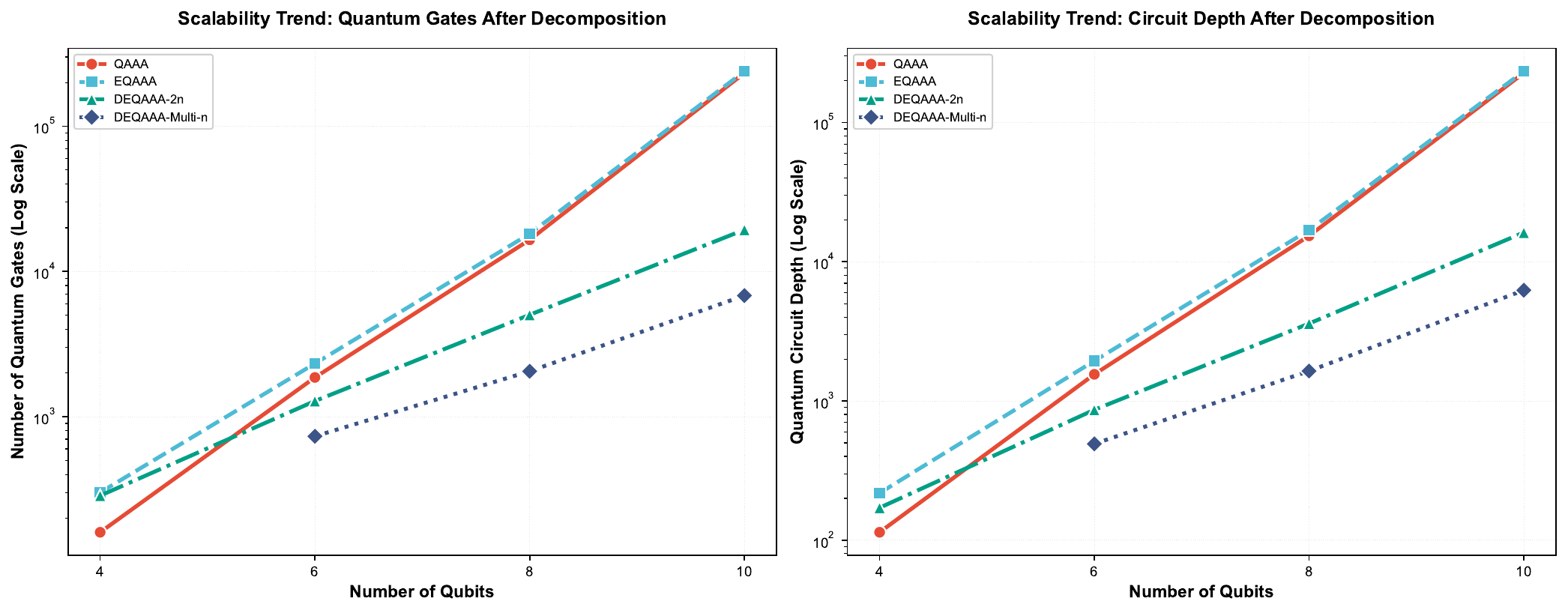}
\vspace{-2.5em}
\caption{Scalability trends of quantum gates and circuit depth of 6-qubit, 8-qubit and 10-qubit examples after decomposition.}
\label{fig-example-compare-trend}
\end{figure}

Overall, the distributed architecture of DEQAAAs grants balanced advantages in hardware adaptability, operational efficiency and resource overhead control, serving as a scalable, practical solution for resource-constrained distributed quantum systems in the NISQ era.

Finally, all simulation experiments and results in this chapter are reproducible, with relevant code available at the open-source URL provided at the end of the paper.

\section{Conclusion}\label{sec-conclusion}
Distributed quantum computing has emerged as a pivotal strategy to transcend the fundamental constraints of monolithic quantum architectures. By partitioning computational tasks across networked quantum processors, this approach not only mitigates the limitations imposed by qubit count and connectivity in individual devices but also establishes a promising pathway toward scalable quantum information processing. Particularly in the NISQ era, the distributed paradigm offers a practical advantage by reducing the circuit depth and qubit requirements at each node, thereby enhancing algorithmic feasibility and error resilience. This architectural shift is poised to unlock larger-scale quantum computation by efficiently leveraging modular quantum resources.

In this study, we address the challenge of achieving exact amplitude amplification for quantum states with arbitrary amplitude distributions by proposing the DEQAAA. Our solution decomposes the problem across an arbitrary number of nodes $t$ (where $2 \leq t \leq n$), enabling a flexible and scalable approach.

More specifically, the DEQAAA exhibits the following key features: (1) it supports flexible partitioning across an arbitrary number of nodes $t$, where $2 \leq t \leq n$, enabling scalable distributed execution; (2) the qubit requirement per node is substantially reduced, with the maximum number needed at any single node being $\max \left(n_0,n_1,\dots,n_{t-1} \right) $, where $n_j$ represents the qubit count at the $j$-th node and the constraint $\sum_{j=0}^{t-1} n_j =n$ ensures efficient resource distribution; (3) DEQAAA achieves exact amplitude amplification for multiple targets for a quantum state with an arbitrary amplitude distribution.

To experimentally validate the performance of DEQAAA, we implemented schemes for 4-qubit, 6-qubit, 8-qubit and 10-qubit systems, tailored to solve a specific exact amplitude amplification problem involving two targets (8 and 14 in decimal). This verification was carried out on the MindSpore Quantum simulation platform, which provides a stable environment for quantum circuit and algorithm simulations. The successful execution of the corresponding DEQAAA circuits across these systems confirms their practical viability, showcasing a practical and scalable workflow. A key advantage of the DEQAAA framework lies in distinctively superior quantum gate count and circuit depth relative to the original QAAA, achieved through efficient decomposition of the $C^{n-1}PS$ gate. In the 10-qubit scenario, for example, both indicators are reduced by over $97\%$ compared to QAAA and EQAAA, highlighting its remarkable resource efficiency. This low-depth architecture significantly enhances noise resilience of the algorithm, a critical feature for practical application in the NISQ era.

In conclusion, the DEQAAA overcomes critical limitations of existing QAAA, emerging as a powerful tool for distributed quantum computing in the NISQ era. Our results highlight the potential of distributed quantum algorithms to surmount the fundamental constraints of single-device architectures, thereby providing a more practical and scalable approach to quantum information processing within the current technological landscape. We anticipate that DEQAAA will play a pivotal role in the development of scalable and reliable quantum computing systems.

The two unresolved issues mentioned at the end of subsection \ref{circuitdepth} will be left for further investigation and resolution in future research.

\section*{Declaration of competing interest}
The authors declare that they possess no conflicting financial interests or personal relationships that could potentially bias the research findings presented in this paper.

\section*{Data availability statement}
The data that support the findings of this study are openly available at the following URL/DOI: \url{https://gitee.com/zhou-xu3/deqaaa\_code}.

\section*{Acknowledgements}
\noindent This work is supported by the China Postdoctoral Science Foundation under Grant No.2023M740874, the Scientific Foundation for Youth Scholars of Shenzhen University, Guangdong Provincial Quantum Science Strategic Initiative under Grant No.GDZX2403001 and No.GDZX2303001, the National Natural Science Foundation of China under Grant No.62571343, the Guangdong Provincial Quantum Science Strategic Initiative under Grants No.GDZX2403001 and the Quantum Science and Technology-National Science and Technology Major Project under Grant No.2021ZD0302901. Sponsored by CPS-Yangtze Delta Region Industrial Innovation Center of Quantum and Information Technology-MindSpore Quantum Open Fund.

\section*{Appendixes}

\begin{appendices}
\setcounter{equation}{0}
\renewcommand\theequation{A\arabic{equation}} 
\setcounter{figure}{0}
\renewcommand\thefigure{\Alph{section}\arabic{figure}}

\section{Analysis of the QAAA in subsection \ref{sec-QAA}}\label{AnalysisQAAA}
After applying the unitary operator $\mathcal{A}$, we obtain
\begin{eqnarray}
\vert\Psi\rangle &=& \mathcal{A} \vert 0\rangle^{\otimes n} =  \sum_{x \in \{0,1\}^n} \gamma_x \vert x\rangle \\
&=& \sum_{x \in X_g} \alpha_x \vert x\rangle + \sum_{x \in X_b} \beta_x \vert x\rangle\\
&=& \sqrt{p_g} \sum_{x \in X_g} \frac{\alpha_x}{\sqrt{p_g}} \vert x\rangle + \sqrt{1-p_g} \sum_{x \in X_b}  \frac{\beta_x}{\sqrt{1-p_g}} \vert x\rangle\\
&=&\sqrt{p_g}\vert\Psi_g\rangle + \sqrt{1-p_g}\vert\Psi_b\rangle,
\end{eqnarray}
where
\begin{eqnarray}
\sin \theta = \sqrt{p_g},~\cos \theta = \sqrt{1-p_g},
\end{eqnarray}
and
\begin{equation}
p_g = \sum_{x \in X_g} \vert\alpha_x\vert^2.
\end{equation}
Here, the states $\vert\Psi_g\rangle$ and $\vert\Psi_b\rangle $ are defined as
\begin{eqnarray}
\vert\Psi_g\rangle &=& \sum_{x \in X_g} \frac{\alpha_x}{\sqrt{p_g}} \vert x\rangle,\\
\vert\Psi_b\rangle &=& \sum_{x \in X_b}  \frac{\beta_x}{\sqrt{1-p_g}} \vert x\rangle.
\end{eqnarray}

Similar to Grover's algorithm, each application of $Q$ induces a rotation by an angle $2\theta$, moving the state closer from the initial state $\vert\Psi\rangle$ to the target state $\vert\Psi_g\rangle$, as illustrated in Figure \ref{QAAArotationG}.
\begin{figure}[H]
\centering
\includegraphics[width=0.35\textwidth]{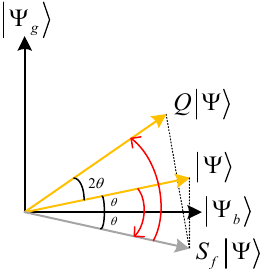}
\caption{Schematic illustration of the state evolution under the $Q$ operator.}
\label{QAAArotationG}
\end{figure}

After $r$ applications of $Q$, the state evolves to
\begin{eqnarray}
\vert\Psi'\rangle  = Q^{r} \vert\Psi\rangle = \sin\left((2r+1)\theta\right)\vert\Psi_g\rangle + \cos\left((2r+1)\theta\right)\vert\Psi_b\rangle.
\end{eqnarray}
To ensure
\begin{eqnarray}
\sin((2r+1)\theta) \approx 1,
\end{eqnarray}
we require
\begin{eqnarray}
(2r+1)\theta \approx \frac{\pi}{2},
\end{eqnarray}
leading to
\begin{eqnarray}
r \approx \frac{\pi/2-\theta}{2\theta}=\frac{\pi}{4\theta}-\frac{1}{2}.
\end{eqnarray}
Since $r$ must be an integer, we choose
\begin{eqnarray}
r=\left\lfloor\frac{\pi}{4\arcsin\left(\sqrt{p_g}\right)}\right\rfloor.
\end{eqnarray}
This choice maximizes the success probability, which is given by
\begin{eqnarray}
\sin^2 \left( \left(2\left\lfloor\frac{\pi}{4\arcsin\left(\sqrt{p_g}\right)}\right\rfloor+1\right) \arcsin \left(\sqrt{p_g}\right) \right) \approx 1.
\end{eqnarray}

\setcounter{equation}{0}
\renewcommand\theequation{B\arabic{equation}} 
\setcounter{figure}{0}
\renewcommand\thefigure{\Alph{section}\arabic{figure}}

\section{Analysis of the EQAAA in subsection \ref{sec-EQAA}}\label{AnalysisEQAAA}
The exact amplitude amplification operator $EQ$ admits the following equivalent representations:
\begin{eqnarray}
EQ&=& \mathcal{A} R_{\vert 0 \rangle ^ {\otimes n}} ^{\phi} \mathcal{A}^{\dagger} R_{f}^{\phi}\\
  &=& \left(I^{\otimes n} + \left(e^{i\phi}-1\right)\vert \Psi\rangle\langle \Psi\vert \right) \left(I^{\otimes n} + \left(e^{i\phi}-1\right)\vert \Psi_g \rangle\langle \Psi_g \vert  \right)\\
&=& e^{i\phi} \left[ \cos\left( \frac{\alpha}{2} \right)I + i\sin\left( \frac{\alpha}{2} \right) (n_x X + n_y Y + n_z Z)  \right] \label{laxis},
\end{eqnarray}
where the angle $\alpha$ and the components of the rotation axis are defined by
\begin{eqnarray}\label{alpha}
\alpha = 4\beta,~\sin \beta = \sin \left(\frac{\phi}{2}\right) \sin \theta,~\sin \theta = \sqrt{p_g},
\end{eqnarray}
and
\begin{eqnarray}\label{nxnynz}
n_x = \frac{\cos \theta}{\cos \beta} \cos \left(\frac{\phi}{2}\right), n_y = \frac{\cos \theta}{\cos \beta} \sin \left(\frac{\phi}{2}\right), n_z = \frac{\cos \theta}{\cos \beta} \cos \left(\frac{\phi}{2}\right) \tan \theta.
\end{eqnarray}
Here, $I$, $X$, $Y$, $Z$ denote the Pauli gates.

Within the three-dimensional space spanned by $\{\vert \Psi_g\rangle, \vert \Psi_b\rangle\}$, each application of $EQ$ can be viewed as a rotation by an angle $\alpha$ around the axis $\vec{l}= \left[n_x, n_y, n_z \right]^T$, which gradually evolves the state from  $\vert \Psi \rangle$ toward $\vert \Psi_g \rangle$. The total rotation angle is denoted by $\omega$ (see Figure \ref{EQAAArotationEQ}), which satisfies
\begin{eqnarray}
\cos \omega  = \cos \left( 2 \arccos \left( \sin \left(\frac{\phi}{2}\right) \sqrt{p_g} \right) \right).
\end{eqnarray}
It follows that
\begin{eqnarray}\label{omega}
\omega  = 2 \arccos \left( \sin \left(\frac{\phi}{2}\right) \sqrt{p_g} \right) = 2 \left( \frac{\pi}{2} -  \arcsin \left( \sin \left(\frac{\phi}{2}\right) \sqrt{p_g} \right) \right).
\end{eqnarray}

\begin{figure}[H]
\centering
\includegraphics[width=0.3\textwidth]{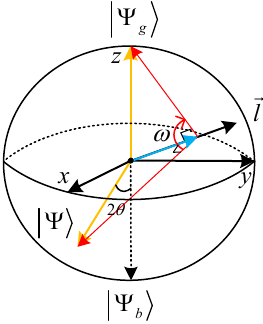}
\caption{Schematic illustration of the state evolution under the $EQ$ operator.}
\label{EQAAArotationEQ}
\end{figure}

After $J+1$ applications of $EQ$, to ensure the desired state $\vert \Psi_g \rangle$ is obtained with a theoretical success probability of $100\%$, the angle $\omega$ must satisfy
\begin{eqnarray}\label{omega2}
\omega  =  (J+1)\alpha = 4 (J+1) \arcsin \left( \sin \left(\frac{\phi}{2}\right) \sqrt{p_g} \right).
\end{eqnarray}
Equating Eq.~\eqref{omega} and Eq.~\eqref{omega2} yields the condition
\begin{eqnarray}\label{condition}
\sin \left( \frac{\pi}{4J+6} \right) = \sin \left(\frac{\phi}{2}\right) \sqrt{p_g}.
\end{eqnarray}
Consequently, the phase angle $\phi$ is given by
\begin{equation}
    \phi = 2 \arcsin \left(\frac{\sin \left(\frac{\pi}{4 J + 6}\right)}{\sqrt{p_g}}\right).
\end{equation}

\setcounter{equation}{0}
\renewcommand\theequation{C\arabic{equation}} 
\setcounter{figure}{0}
\renewcommand\thefigure{\Alph{section}\arabic{figure}}

\section{Proof of Theorem \ref{proofcorrectness}}\label{prooftheorem1}
The correctness of DEQAAA relies on the exact amplitude amplification property of EQAAA (detailed in Appendix \ref{AnalysisEQAAA}). Specifically, the core operator $EQ$ of EQAAA achieves $100\%$ target probability through rotational evolution in the subspace spanned by the global target state $\vert\Psi_g\rangle$ and the global non-target state $\vert\Psi_b\rangle$. DEQAAA implements exact amplification in two phases, and the correctness analysis is conducted based on the value of $p'_g$ (the global success probability after Phase 1), as follows:

\textbf{Case 1:} \bm{$p'_g = 1$}

Case 1 corresponds to $p'_g = 1$, where $p'_g = \sum_{x \in X_g} P'(x)$ denotes the global success probability after Phase 1, and $P'$ is the exact probability distribution of the state $\vert\Psi_1\rangle = \left(\bigotimes_{j=0}^{t-1} EQ_j^{J_j+1}\right)\mathcal{A}\vert 0\rangle^{\otimes n}$. When $p'_g = 1$, $\vert\Psi_1\rangle$ already satisfies $\sum_{x \in X_g} \vert\langle x \vert \Psi_1 \rangle \vert^2 = 1$, indicating that Phase 1 alone achieves global exact amplification. In this case, DEQAAA omits Phase 2, and the final state $\vert\Psi_2\rangle = \vert\Psi_1\rangle=\vert\Psi_g\rangle$ directly meets the requirement of $\sum_{x \in X_g} \vert\langle x \vert \Psi_2 \rangle \vert^2 = 1$.

\textbf{Case 2:} \bm{$p'_g \neq 1$}

Case 2 corresponds to $p'_g \neq 1$, where Phase 1 fails to amplify the success probability of $\vert\Psi_1\rangle$ to 1 via parallel local exact amplitude amplification. To achieve global exact amplification, DEQAAA proceeds to Phase 2, first encapsulating the complete unitary operation of Phase 1 into a composite operator $\mathcal{B} = \left(\bigotimes_{j=0}^{t-1} EQ_j^{J_j+1}\right)\mathcal{A}$. The core of Phase 2 is the global exact amplitude amplification operator $\widehat{EQ} = \mathcal{B}R_{\vert 0\rangle^{\otimes n}}^{\hat{\phi}}\mathcal{B}^{\dagger}R_f^{\hat{\phi}}$, whose form is fully consistent with EQAAA’s core operator $EQ$.

By setting the phase angle $\hat{\phi}$ and iteration count $\hat{J}$ to satisfy EQAAA’s rotational angle condition (Eq.~\eqref{condition}), $\hat{J}+1$ iterations of $\widehat{EQ}$ can exactly evolve the state $\vert\Psi_1\rangle$ to $\vert\Psi_g\rangle = \vert\Psi_2\rangle$. Here, $\hat{J}= \left\lfloor \frac{\pi}{4\arcsin\left(\sqrt{p'_g}\right)} - \frac{1}{2} \right\rfloor$ is determined by $p'_g$, ensuring the rotational evolution meets the exact amplification requirement.

The ``Phase 1 + Phase 2" process is equivalent to the standard EQAAA process with $\mathcal{B}$ as the state preparation operator, whose correctness is guaranteed by EQAAA’s exact amplification property. The final state $\vert\Psi_2\rangle = \vert\Psi_g\rangle$ thus satisfies $\sum_{x \in X_g} \vert\langle x \vert \Psi_2 \rangle \vert^2 = 1$.

In summary, regardless of whether $p'_g$ is 1 or not, DEQAAA can achieve exact amplitude amplification for the target strings $x \in X_g$, and the final state $\vert\Psi_2\rangle$ satisfies $\sum_{x \in X_g} \vert\langle x \vert \Psi_2 \rangle \vert^2 = 1$. Thus, the correctness of DEQAAA is proven.

\setcounter{equation}{0}
\renewcommand\theequation{D\arabic{equation}} 
\setcounter{figure}{0}
\renewcommand\thefigure{\Alph{section}\arabic{figure}}

\section{Proof of Theorem \ref{theorem2qaaadepth}}\label{prooftheorem2}
The unitary operator $\mathcal{A}$ prepares an $n$-qubit state $\vert\Psi\rangle$ with an arbitrary amplitude distribution, and its implementation depth is defined as $\mathrm{dep}(\mathcal{A})$. 

Without loss of generality, let $p_g$ be the probability of measuring a target string in the global target set $X_g$, and $\vert X_g\vert $ the number of target strings in $X_g$. 

In QAAA, the amplitude amplification operator is defined as $Q=\mathcal{A}S_{\vert 0\rangle^{\otimes n}}\mathcal{A}^\dagger S_f$, which is executed repeatedly $r$ times, where $r$ is given by: 
\begin{equation}
    r=\left\lfloor\frac{\pi}{4\arcsin\left(\sqrt{p_g}\right)}\right\rfloor.
\end{equation}

For simplicity, the circuit depth of the $C^{n-1}Z$ gate is defined as 1. From this, the depth of $S_f$ (for $\vert X_g \vert$ targets) is $\text{dep}(S_f)=3\vert X_g\vert $, and the depth of $S_{\vert 0\rangle^{\otimes n}}$ is $\text{dep}(S_{\vert 0\rangle^{\otimes n}})=3$. 

The total depth of QAAA is the sum of the initial state preparation depth and the depth of $r$ repeated $Q$ operations. Thus, the total depth is
\begin{eqnarray}
\text{dep}(\text{QAAA}, p_g, \vert X_g\vert) 
&=& \text{dep}(\mathcal{A}) + r \cdot \left(3\vert X_g\vert + 2\text{dep}(\mathcal{A}) + 3\right) \\
&=& (2r + 1) \cdot \text{dep}(\mathcal{A}) + r \cdot (3\vert X_g\vert + 3).
\end{eqnarray}

This completes the proof.

\setcounter{equation}{0}
\renewcommand\theequation{E\arabic{equation}} 
\setcounter{figure}{0}
\renewcommand\thefigure{\Alph{section}\arabic{figure}}

\section{Proof of Theorem \ref{theorem3eqaaadepth}}\label{prooftheorem3}
Following the basic definitions in Theorem \ref{theorem2qaaadepth}, the implementation depth of the unitary operator $\mathcal{A}$ is $\text{dep}(\mathcal{A})$; $p_g$ denotes the probability of measuring a target string in the global target set $X_g$, and $\vert X_g\vert$ is the number of target strings in $X_g$. 

In EQAAA, the exact amplitude amplification operator is defined as $EQ=\mathcal{A} R_{\vert 0\rangle^{\otimes n}}^\phi \mathcal{A}^\dagger R_f^\phi$, which is iteratively applied $J+1$ times, with $J$ given by: 
\begin{equation}
    J = \left\lfloor \frac{\pi}{4\arcsin\left(\sqrt{p_g}\right)} - \frac{1}{2} \right\rfloor.
\end{equation}

Similar to the $C^{n-1}Z$ gate, the circuit depth of the $C^{n-1}\text{PS}(\phi)$ gate is defined as 1 for simplicity. From this, the depth of $R_f^\phi$ (for $\vert X_g\vert$ targets) is $\text{dep}(R_f^\phi)=3\vert X_g\vert$, and the depth of $R_{\vert 0\rangle^{\otimes n}}^\phi$ is $\text{dep}(R_{\vert 0\rangle^{\otimes n}}^\phi)=3$, which is consistent with the depths of $S_f$ and $S_0$. 

The total depth of EQAAA is the sum of the initial state preparation depth and the depth of $J+1$ repeated $EQ$ operations. Thus, the total depth is
\begin{eqnarray}
\text{dep}(\text{EQAAA}, p_g, \vert X_g\vert)
&=& \text{dep}(\mathcal{A}) + \left( J +1 \right) \cdot \left(3\vert X_g\vert+2\text{dep}(\mathcal{A})+3\right) \\
&=&  \left( 2J +3 \right) \cdot \text{dep}(\mathcal{A}) + \left( J+1 \right) \cdot \left(3\vert X_g\vert+3\right).
\end{eqnarray}

This completes the proof.

\setcounter{equation}{0}
\renewcommand\theequation{F\arabic{equation}} 
\setcounter{figure}{0}
\renewcommand\thefigure{\Alph{section}\arabic{figure}}

\section{Proof of Theorem \ref{theorem4deqaaadepth}}\label{prooftheorem4}
Following the basic definitions in Theorem \ref{theorem2qaaadepth} and Theorem \ref{theorem3eqaaadepth}, the global target set is $X_g$ (with $\vert X_g\vert$ target strings), and the distributed architecture contains $t$ nodes; the $j$-th node has a local target set $X_j$ (with $\vert X_j\vert$ target strings), a success measurement probability $p_j$, and the implementation depth of the corresponding local unitary preparation operator $\mathcal{A}_{\vert \varphi_j \rangle}$ (generating substate $\vert\varphi_j\rangle = \mathcal{A}_{\vert \varphi_j \rangle}\vert 0\rangle^{\otimes n_j}$) is $\text{dep}(\mathcal{A}_{\vert \varphi_j \rangle})$. The global initial state is prepared by the unitary operator $\mathcal{A}$, with depth $\text{dep}(\mathcal{A})$.

The subsequent proof will derive the circuit depths of the first phase (local execution of distributed nodes) and the second phase (global exact amplification) respectively, and finally integrate to obtain the total depth of DEQAAA.

\textbf{Phase 1: Local Execution of Distributed Nodes}

The first phase is the local execution phase of distributed nodes, where each node independently performs local EQAAA operations. The core is the dedicated exact amplitude amplification operator $EQ_j$ (defined in Eq.~\eqref{eq-EQ-j}), which consists of local phase rotation operators $R_{f_j}^{\phi_j}$, $R_{\vert 0 \rangle ^ {\otimes n_j}} ^{\phi_j}$, and $\mathcal{A}_{\vert \varphi_j \rangle}$ and its inverse, with the phase angle $\phi_j$ calculated by Eq.~\eqref{eq-angle-phi-j}.

The $EQ_j$ operator is iteratively applied $J_j + 1$ times, with the iteration count $J_j$ given by 
\begin{eqnarray}
J_j &=& \left\lfloor \frac{\pi}{4\arcsin\left(\sqrt{p_j}\right)} - \frac{1}{2} \right\rfloor, 
\end{eqnarray}
consistent with the global iteration logic of EQAAA. Referring to the depth definition logic in Theorem \ref{theorem3eqaaadepth}, the depth of the $C^{n-1}\text{PS}(\phi)$ gate is defined as 1. Since the local phase rotation operators $R_{f_j}^{\phi_j}$ and $R_{\vert 0\rangle^{\otimes n_j}}^{\phi_j}$ are both constructed based on the $C^{n_j-1}\text{PS}(\phi_j)$ gate, their depths are both set to 1. Thus, the operation depth of the $j$-th node is 
\begin{eqnarray}
\text{dep}(\text{$j$-th node, $p_j$, $\vert X_j\vert$}) = \left( J_j +1 \right) \cdot \left(3\vert X_j\vert + 2\text{dep}(\mathcal{A}_{\vert \varphi_j \rangle}) + 3\right),
\end{eqnarray}
which matches the corresponding formula in the theorem. Its depth composition logic is fully consistent with that of the global operator in EQAAA, i.e., the total depth is the superposition of $J_j+1$ iterations of the single $EQ_j$ operator.

As each node performs local operations in parallel, the total depth of the first phase is the sum of the global initial state preparation depth and the maximum operation depth among all nodes, i.e.,
\begin{eqnarray}
\text{dep}(\text{First phase}) = \text{dep}(\mathcal{A}) + \max\limits_{j \in \{0,1,\cdots,t-1 \}} \left\{ \text{dep}(\text{$j$-th node, $p_j$, $\vert X_j\vert$})\right\}.
\end{eqnarray}

\textbf{Phase 2: Global Exact Amplification}

After the first phase, the system state evolves to $\vert \Psi_1 \rangle$, and its global success probability (i.e., the success probability of the second phase) is defined as 
\begin{equation}
p'_g= \sum_{x \in X_g} P'(x) ,
\end{equation}
where $P'$ is the exact probability distribution corresponding to $\vert \Psi_1 \rangle$. The second phase is global exact amplification, executed only when $p'_g \neq 1$. In this case, the first phase cannot amplify the success probability of $\vert\Psi_1\rangle$ to 1 (failing to achieve global exact amplification), requiring correction via global exact amplitude amplification operations.

In this phase, the complete unitary operation of the first phase is first encapsulated into a composite operator 
\begin{eqnarray}
\mathcal{B} = \left[ \bigotimes_{j=0}^{t-1} \left( EQ_{j}^{J_{j}+1} \right)  \right]  \mathcal{A},
\end{eqnarray}
with the core being the global exact amplitude amplification operator 
\begin{equation}
    \widehat{EQ}= \mathcal{B} R_{\vert 0 \rangle ^{\otimes n}} ^{\hat{\phi}} \mathcal{B}^{\dagger} R_{f}^{\hat{\phi}}.
\end{equation}
The iteration count 
\begin{equation}
    \hat{J}= \left\lfloor \frac{\pi}{4\arcsin\left(\sqrt{p'_g}\right)} - \frac{1}{2} \right\rfloor,
\end{equation}
which is consistent with the iteration logic of $J$ in Theorem \ref{theorem3eqaaadepth}.

Referring to the depth derivation logic of Theorem \ref{theorem3eqaaadepth}, $\text{dep}(\mathcal{B}) = \text{dep}(\text{First phase})$ and $\text{dep}(\mathcal{B}^\dagger) = \text{dep}(\mathcal{B})$ (the depth of an operator is equal to that of its inverse). The depths of global phase rotation operators are consistent with EQAAA: the depth of $R_f^{\hat{\phi}}$ is $3\vert X_g\vert$, and the depth of $R_{\vert 0\rangle^{\otimes n}}^{\hat{\phi}}$ is 3. Thus, the depth of a single $\widehat{EQ}$ iteration is $3\vert X_g\vert + 2\text{dep}(\text{First phase}) + 3$, and the total depth of the second phase is 
\begin{eqnarray}
\text{dep}(\text{Second phase})=\left( \hat{J} +1 \right) \cdot \left(3\vert X_g\vert + 2\text{dep}(\text{First phase}) + 3\right).
\end{eqnarray}

In summary, the total circuit depth of DEQAAA is determined by the value of $p'_g$: 

When $p'_g = 1$, the local exact amplitude amplification executed in parallel in the first phase suffices to achieve exact amplification of the global target state, and the total depth is the depth of the first phase, i.e., 
\begin{eqnarray}
\text{dep}(\text{DEQAAA}) = \text{dep}(\text{First phase}).
\end{eqnarray}

When $p'_g \neq 1$, global exact amplification in the second phase is required to correct deviations, and the total depth is the sum of the depths of the two phases, i.e., 
\begin{eqnarray}
\text{dep}(\text{DEQAAA}) = \text{dep}(\text{First phase}) + \text{dep}(\text{Second phase}).
\end{eqnarray}

Thus, the circuit depth of DEQAAA satisfies Eq.~\eqref{depthdeqaaa}, and the theorem is proven.

\setcounter{equation}{0}
\renewcommand\theequation{F\arabic{equation}} 
\setcounter{figure}{0}
\renewcommand\thefigure{\Alph{section}\arabic{figure}}
\setcounter{lemma}{0}
\renewcommand\thelemma{\Alph{section}\arabic{lemma}}
\section{The decomposition of $C^{3}PS(\phi)$ in subsection \ref{sec-datacomparison}}\label{appx-decom}
In the above experiments, QAAA achieves phase flip using $C^{3}Z$ gates, whereas EQAAA and DEQAAA realize phase rotation through $C^{3}PS(\phi)$ gates. Note that the $C^{3}Z$ gate is a special case of the $C^{3}PS(\phi)$ gate with $\phi = \pi$, i.e., $C^{3}Z = C^{3}PS(\pi)$.

However, current quantum computers generally do not support the direct implementation of such complex multi-controlled gates. A more common approach is to decompose them into equivalent combinations of single-qubit and two-qubit gates.

The decomposition of the $C^{3}PS(\phi)$ gate was also discussed in Ref.~\cite{zhou2025distributedEGGA}. Nevertheless, this method requires the execution of multiple CPS and CNOT gates. Since two-qubit gates like these are inherently more challenging to implement with high fidelity than single-qubit operations, they significantly increase the practical complexity of the circuit. Furthermore, as the number of qubits $n$ increases, this approach adopts the technique outlined in Lemma 7.5 of Ref.~\cite{barenco1995elementary}. A notable drawback of this method is that it further necessitates the decomposition of multi-controlled $X$ gates, thereby introducing an additional layer of complexity and resource overhead.

To enhance the implementation efficiency of multi-controlled phase gates on NISQ devices, we propose an improved decomposition scheme. This scheme constructs a deterministic quantum circuit that expresses the gate as an optimized sequence of elementary single-qubit gates and CNOT gates.

Specifically, we present an enhanced decomposition scheme that implements the $C^{3}PS(\phi)$ gate. This method naturally generalizes to arbitrary $n$-qubit controlled-phase gates ($C^{n-1}PS(\phi)$), though the detailed derivation is not included in this paper.

\begin{lemma}\label{theorem-CnPS-decom}
(\textbf{Decomposition of $C^{3}PS(\phi)$}) For any single-qubit unitary gate $PS(\phi)$, the $C^{3}PS(\phi)$ gate can be decomposed into the quantum circuit depicted in Figure \ref{fig-C3PS-decom-circuit}, where $\phi$ denotes the phase angle and $n=4$ is the total number of qubits involved in the $C^{3}PS(\phi)$ gate.

\begin{figure}[H]
\centering
\includegraphics[width=\textwidth]{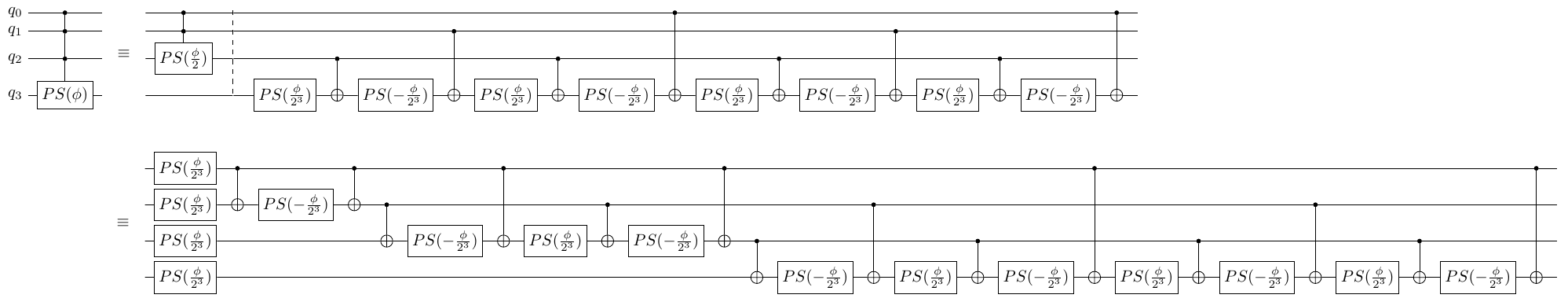}
\vspace{-1.8em}
\caption{Quantum circuit for the decomposition of $C^3PS(\phi)$ gate.}
\label{fig-C3PS-decom-circuit}
\end{figure}

\end{lemma}

\begin{proof}
It is straightforward to confirm that the unitary matrices associated with the quantum circuits on both sides of the equality sign are equivalent.
\end{proof}

\end{appendices}

\balance

\newpage
\bibliography{References}
\newpage

\end{document}